\documentclass[10pt,journal]{IEEEtran}

\usepackage{amssymb}
\usepackage{amsmath}
\usepackage{cite}
\usepackage{url}
\usepackage{xcolor}
\usepackage{cite,graphicx,amsmath,amssymb}
\usepackage{fancyhdr}
\usepackage{mdwmath}
\usepackage{mdwtab}
\usepackage{caption}
\usepackage{amsthm}
\usepackage{setspace}
\usepackage{hyperref}
\usepackage{algorithm}
\usepackage{algorithmic}
\usepackage{multirow}
\usepackage{makecell}
\usepackage{mathtools}
\usepackage{subcaption}
\usepackage{bm}
\usepackage{tikz}
\usepackage{xurl}
\usepackage{stfloats}
\usepackage{mathrsfs}

\usepackage{booktabs}
\usepackage{multirow}
\usepackage{siunitx}
\usepackage{balance}

\hypersetup{colorlinks=true,
linkcolor=blue,
citecolor=blue,      
urlcolor=black,
}

\graphicspath{{./src/img/}}

\DeclareMathOperator*{\argmax}{arg\,max}

\newtheorem{remark}{Remark}
\newtheorem{theorem}{Theorem}

\newtheorem{lemma}{Lemma}

\newtheorem{corollary}{Corollary}

\newtheorem{proposition}{Proposition}

\allowdisplaybreaks
\NewDocumentCommand{\multiubrace}{mmm}
 {
  \egreg_multiubrace:nnn {#1} {#2} {#3}
 }

\newcommand{\softmax}{\mathrm{softmax}}
\newcommand{\kron}{\otimes}

\title{Generative Site-Specific Beamforming for UPAs via Decoupled Channel Sensing}

\author{
        Yao Tang and Zhaolin Wang 
\thanks{
Yao Tang is with the AI for Science and Engineering Center, Shenzhen Loop Area Institution, China (e-mail: yaotang@slai.edu.cn).}
\thanks{
Zhaolin Wang is with the Department of Electrical and Computer Engineering, The University of Hong Kong, Hong Kong (e-mail: zhaolin.wang@hku.hk).}
}

\begin{document}

\maketitle

\begin{abstract}
A cross-fused generative site-specific beamforming (GenSSBF) framework is proposed for low-overhead beam alignment in uniform planar array (UPA) systems. A decoupled channel sensing strategy is developed, where the azimuth and elevation domains of the UPA are probed independently, and the online sweeping overhead is reduced from multiplicative to linear complexity compared to exhaustive two-dimensional codebook sweeping. 
However, the resulting reference signal received power (RSRP) observations only contain marginal angular power information. The explicit azimuth-elevation coupling of the UPA channel is therefore lost. Beam generation from these separate observations becomes highly ambiguous. 
To address this issue, a bidirectional cross-attention encoder is designed to extract and fuse the latent dependency between the azimuth and elevation sensing branches. 
Conditioned on the fused feature, a conditional normalizing flow generator is proposed to generate a compact set of high-fidelity beam candidates. These candidates are further verified through lightweight pilot measurements for final beam selection. 
A task-oriented training objective is also introduced to encourage the generated candidate set to contain at least one high-gain beam, rather than fitting the full conditional beam distribution. 
Simulation results based on DeepMIMO scenarios show that the proposed framework consistently outperforms deterministic beam prediction and conventional discrete Fourier transform (DFT) codebook search. Compared with the full 1024-beam two-dimensional DFT search, normalized beamforming gain improvements of 83.6\%, 74.6\%, and 38.1\% are achieved in the I2\_28, O1B\_28, and Boston5G\_28 scenarios, respectively, while the sweeping overhead is reduced by 93.8\%.
\end{abstract}

\begin{IEEEkeywords}
    Site-specific beamforming, normalizing flow model, generative artificial intelligence.
\end{IEEEkeywords}

\section{Introduction}


\IEEEPARstart{B}EAMFORMING has become a fundamental technique in modern wireless networks, enabling base stations (BSs) equipped with multiple antennas to exploit spatial degrees of freedom for coverage enhancement, interference suppression, and spectral efficiency improvement \cite{MIMO, xue2024survey, ning2023beamforming}. Its importance spans both sub-6 GHz massive multiple-input multiple-output (MIMO) systems and higher-frequency deployments, where directional transmission is needed to support reliable and high-rate wireless links \cite{mmWave_ming, mmWave_kutty, heath2016overview}. By concentrating radiated energy toward intended user equipments (UEs), beamforming can substantially improve the received signal strength and link quality. However, these gains critically depend on accurate beam alignment between the BS and UE, which remains a key challenge as antenna arrays become larger and wireless channels become more spatially selective.

In practical fifth-generation new radio (5G NR) systems, beam alignment is commonly realized through channel-dependent beamforming or grid-of-beams based sweeping \cite{alkhateeb2014channel, New_5g}. The former can achieve near-optimal performance. However, it relies on accurate channel state information (CSI), leading to substantial pilot training, feedback, and signal processing overhead, particularly in large-scale antenna systems. The latter avoids explicit full CSI acquisition by sweeping a predefined codebook and selecting the beam with the strongest reference signal received power (RSRP) \cite{dahlman20205g,song2017common}. This procedure reduces signal processing complexity and feedback granularity, but its adaptivity is limited by the finite beam set. 
Consequently, existing beam alignment methods face a fundamental tradeoff between beamforming fidelity and online sweeping overhead.

To address this, site-specific beamforming (SSBF) has recently emerged as a promising direction \cite{li2019learning_codebook, kwak2024sitespecific_codebook, heng2022learning_codebook, dreifuerst2025neural, ning2023learning, heng2024_gridfree, abdallah2025explainable, wu2024environmentaware, alkhateeb2018deep}.The key idea of SSBF is to exploit the fact that wireless propagation around a BS is not arbitrary, but is shaped by persistent site-dependent factors such as building geometry, dominant reflectors, blockage regions, and user distributions. By learning such environmental regularities offline, the BS can infer high-quality beams using deep neural networks from a small number of online sensing measurements. More recently, generative SSBF (GenSSBF) has further shown that beamforming should not always be treated as a deterministic prediction problem \cite{zhaolin2026, zihao2026beam, zihao2026fastbeam, chengjie2026optcodebook, chengjie2026learncodebook}. The coarse measurements such as RSRP discard phase information and may correspond to multiple plausible channels. Therefore, a generative model can produce a set of candidate beams and thereby better handle the intrinsic multimodality of beam inference.

Existing GenSSBF studies have mainly focused on beam generation from low-dimensional probing feedback obtained from one-dimensional probing feedback or generic codebook measurements. For example, diffusion-based beam-brainstorm generates multiple beam candidates from RSRP prompts \cite{zihao2026beam}, fast beam-brainstorm improves the generation speed and supports flexible probing lengths \cite{zihao2026fastbeam}, information-maximizing codebook design improves the informativeness of the probing stage \cite{chengjie2026optcodebook}, and site-specific propagation priors were embedded into a standardized limited-feedback framework by using an RSRP fingerprint to infer a dominant beam subspace before low-dimensional coefficient feedback \cite{chengjie2026learncodebook}. 
These studies demonstrate the potential of generative beam inference, but the overhead issue caused by two-dimensional beam sensing in planar-array systems remains insufficiently explored. Directly extending codebook-based sensing to a full two-dimensional angular grid causes the online sweeping overhead to grow with multiplicative complexity over the two angular dimensions.

A natural way to reduce this overhead is to probe the azimuth and elevation domains separately, which reduces the sensing cost from multiplicative complexity to linear complexity. However, this overhead reduction introduces several new challenges for UPA GenSSBF. 
1) \emph{Information loss}: Decoupled sensing only preserves the marginal power responses along the two angular dimensions, and thus does not explicitly retain the joint azimuth-elevation coupling of the underlying channel. In other words, the BS can observe which azimuth lobes and which elevation lobes are strong, but cannot directly determine which azimuth and elevation lobes correspond to the same physical propagation path. 2) \emph{Lobe mismatch}: Since each one-dimensional probing beam is broad in the other angular dimension, the measured RSRP may contain energy from multiple paths, sidelobe leakage, and reflected components. As a result, simply pairing the strongest azimuth lobe with the strongest elevation lobe may create a virtual angular direction that does not correspond to any true dominant path. 3) \emph{Phase ambiguity}: The power-only nature of RSRP feedback discards the phase and coherent combining information required for high-quality beam synthesis over the full planar aperture. Therefore, different multipath configurations may induce very similar separate RSRP observations while requiring significantly different beamformers.


The above observations suggest that reducing the sensing overhead alone is not sufficient for reliable UPA beam alignment. Since decoupled channel sensing removes the explicit azimuth-elevation coupling, the missing joint angular information must be inferred from the site-specific propagation structure reflected in the separate RSRP observations. 
This motivates the development of a GenSSBF framework that can extract cross-domain dependency from low-dimensional sensing feedback and generate multiple candidate beams under the inherent angular ambiguity.

In this paper, we propose a cross-fused GenSSBF framework for low-overhead beam alignment in UPA-enabled downlink systems. The BS first carries out azimuth and elevation decoupled channel sensing and obtains two low-dimensional RSRP vectors from the UE. Then, a cross-fused encoder learns the latent dependency between the two sensing branches, and a conditional normalizing flow generates a small set of candidate beamformers. Finally, the UE verifies these candidates through lightweight pilot measurements and feeds back the best candidate index for data transmission. We aim to address three main challenges: 1) How can the online channel sensing overhead of UPA beam alignment be reduced without full two-dimensional sweeping? 2) How can the latent azimuth-elevation coupling be extracted from separate RSRP observations? 3) How can the generative model be trained to produce task-oriented beam candidates instead of merely fitting the full conditional density?

The main contributions are summarized as follows:

\begin{itemize}
    \item \textit{Low-overhead decoupled sensing for UPA:} We develop a decoupled channel sensing framework for UPA downlink systems, where the azimuth and elevation domains are probed independently. This reduces the online sensing overhead from multiplicative complexity to linear complexity. We further characterize the key limitation of decoupled sensing, namely that the resulting RSRP observations only contain marginal angular power information and therefore lose the explicit joint azimuth-elevation coupling. This explains why direct deterministic inversion from separate RSRP observations to the beamformer is generally ambiguous.

    \item \textit{Cross-fused GenSSBF architecture:} We propose a generative beamforming framework that consists of angular-domain channel encoding, a cross-fused encoder, and a conditional normalizing flow generator. The angular-domain channel is represented by amplitude and cosine-sine phase components to avoid phase wrapping. The cross-attention encoder allows the azimuth and elevation sensing branches to exchange complementary information and recover latent cross-domain dependency. Conditioned on the fused feature, the normalizing flow generator samples multiple feasible beam candidates from a simple Gaussian latent distribution.

    \item \textit{Task-oriented training loss:} We design a normalized beamforming gain loss that directly optimizes the best candidate among multiple generated beams. Different from likelihood-based training, which attempts to fit the entire conditional beam distribution, the proposed objective focuses on whether the generated candidate set covers at least one high-gain beam. We also provide a high-gain coverage interpretation, showing that the success probability increases with the probability mass assigned by the generator to the high-gain region and with the number of generated candidates.

    \item \textit{Performance evaluation in realistic propagation scenarios:} We evaluate the proposed framework using DeepMIMO scenarios, including indoor, outdoor blockage, and urban mmWave environments \cite{deepmimo}. Simulation results show that the proposed GenSSBF consistently outperforms deterministic multi-layer perceptron (MLP)-based beam prediction, GenSSBF without cross-attention, and conventional DFT codebook search under the same sensing budget. With a two-dimensional DFT codebook as the full sweeping benchmark, the proposed decoupled sensing design substantially reduces the online sensing overhead while achieving beam patterns and normalized beamforming gains close to the maximum ratio transmission (MRT) upper bound in practical SNR regimes.
\end{itemize}

The rest of this paper is organized as follows. Section \ref{sec:model} introduces the system model, including the channel model and decoupled channel sensing procedure. Section \ref{sec:framework} presents the cross-fused GenSSBF framework, including the angular-domain channel representation, cross-fused encoder, and conditional normalizing flow generator. Section \ref{sec:training_infer} develops the task-oriented training objective and the online candidate verification procedure. Section \ref{sec:simulation} provides simulation results and performance comparisons, and Section \ref{sec:conclusion} concludes the paper.

\emph{Notations:} Scalars, vectors/matrices, and Euclidean subspaces are denoted by regular, boldface, and calligraphic letters, respectively. The sets of complex, real, and integer numbers are represented by $\mathbb{C}$, $\mathbb{R}$, and $\mathbb{Z}$, respectively. The inverse, transpose, and conjugate transpose operations are represented by $(\cdot)^{-1}$, $(\cdot)^T$, and $(\cdot)^H$, respectively. The absolute value and Euclidean norm are indicated by $|\cdot|$ and $\|\cdot\|$, respectively. 

\section{System Model} \label{sec:model}

As shown in Fig.~\ref{fig:system_model}, we consider a narrowband downlink communication system. The BS is equipped with a UPA mounted on the $xy$-plane, with $N_t = N_xN_y$ transmit antennas. The UE is equipped with a single antenna and is located in the far field of the BS array. The downlink channel is assumed to be block fading, i.e., it remains approximately constant within each coherence interval. 
The beam management and feedback procedures considered in this paper are carried out within such a coherence interval.

\subsection{Channel Model}

We adopt a geometric channel model, under which the downlink channel is represented as the superposition of $L$ dominant propagation paths \cite{elayach2014spatially}. Let $\mathbf{h} \in \mathbb{C}^{N_t \times 1}$ denote the instantaneous downlink channel vector, which can be expressed as \cite{heath2018foundations}
\begin{equation}
  \mathbf h = \sum\limits_{l=1}^{L} \alpha_l \, \mathbf a(\phi_l, \theta_l),
  \label{eq:channel}
\end{equation}
where $\alpha_l \in \mathbb{C}$ denotes the complex gain of the $l$-th path, accounting for path loss, shadowing, and small-scale fading, and $(\phi_l,\theta_l)$ denote the corresponding azimuth and elevation angles of departure (AoDs), respectively.

\begin{figure}[t]
\begin{center}
\setlength{\abovecaptionskip}{0cm}
\includegraphics[width=8.8cm, trim = 0.3cm 0cm 0cm 0cm, clip = true]{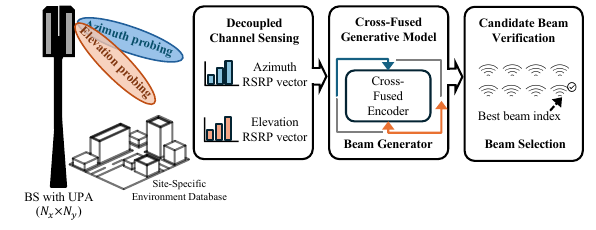}
\caption{\footnotesize{Generative site-specific beamforming (GenSSBF) for a UPA-enabled downlink system.}}
\label{fig:system_model}
\end{center}
\end{figure}

The steering vector $\mathbf a(\phi_l,\theta_l) \in \mathbb{C}^{N_t\times 1}$ characterizes the array response toward direction $(\phi_l,\theta_l)$. For the considered UPA, the steering vector admits a separable Kronecker structure \cite{song2017common}. Specifically, the array response vectors along the $x$- and $y$-axes are given by, respectively,
\begin{equation} \label{equ:az_vector}
\begin{array}{rll}
\!\!\!\! \mathbf a_x(\phi_l, \theta_l) \!=\! \displaystyle \frac{1}{\sqrt{N_x}}
   [1, \!\!\!\!& e^{j\frac{2\pi}{\lambda}d\sin\phi_l\sin\theta_l}, \ldots, \\
& e^{j(N_x-1)\frac{2\pi}{\lambda}d\sin\phi_l\sin\theta_l}]^{T} \!\in\! \mathbb{C}^{N_x\times 1},
\end{array}
\end{equation}
\begin{equation}
\begin{array}{rll}
\mathbf a_y (\theta_l) \!=\! \displaystyle \frac 1 {\sqrt {N_y}} 
[1, \!\!\!\!& e^{j\frac{2\pi}{\lambda}d\cos{\theta_l}}, \ldots, \\
& e^{j(N_y-1)\frac{2\pi}{\lambda}d\cos{\theta_l}}]^{T} \in \mathbb{C}^{N_y\times 1},
\end{array}
\end{equation}
where $\lambda$ and $d$ denote the carrier wavelength and antenna spacing, respectively. Hence, the overall steering vector is
\begin{equation}
  \mathbf a(\phi_l,\theta_l) = \mathbf a_x(\phi_l,\theta_l) \kron \mathbf a_y(\theta_l) \in \mathbb{C}^{N_t\times 1}.
  \label{eq:steer_kron}
\end{equation}

Due to the high hardware cost and power consumption of fully digital beamforming at high carrier frequencies, we consider an analog or hybrid beamforming architecture at the BS \cite{heng2022learning_codebook}. In particular, the BS is assumed to employ a single RF chain together with a phase-shifter network. Let $\mathbf w \in \mathbb{C}^{N_t\times 1}$ denote the analog beamforming vector. It satisfies the transmit power constraint and the constant-modulus constraint imposed by the phase shifters.
Moreover,
let $s$ denote the transmitted symbol intended for the UE, with $\mathbb E[|s|^2]=1$. The received signal at the UE is
\begin{equation}
    y = \sqrt{P_t} \, \mathbf{h}^H \mathbf{w} s + n,
    \label{eq:received_signal}
\end{equation}
where $P_t$ is the transmit power and $n \sim \mathcal{CN}(0,\sigma^2)$ denotes additive white Gaussian noise. Accordingly, the received SNR is given by
\begin{equation}
\text{SNR} = \frac{P_t |\mathbf{h}^H \mathbf{w}|^2}{\sigma^2}.
\label{eq:snr}
\end{equation}

\subsection{RSRP Fingerprint}

In conventional beam management, acquiring the instantaneous full CSI of a large UPA is prohibitively expensive in terms of pilot overhead and feedback signaling. Therefore, practical systems typically perform beam alignment through codebook-based beam sweeping during the initial access stage \cite{giordani2019tutorial,barati2016initial,3gpp38214}. In particular, the BS sequentially transmits synchronization signal blocks (SSBs) over a predefined beam codebook, the UE measures the RSRP associated with each probing beam, and the beam yielding the largest RSRP is selected for access establishment.

Consider a conventional two-dimensional DFT codebook
\begin{equation}
\mathcal C = \{\mathbf c_{m,n}: m=1,\ldots,M_{\phi},\; n=1,\ldots,M_{\theta}\},
\end{equation}
where $M_{\phi}$ and $M_{\theta}$ denote the numbers of candidate beams in the azimuth and elevation domains, respectively. The total number of probing beams is $|\mathcal C| = M_{\phi}M_{\theta}$. For the probing beam $\mathbf c_{m,n}$, the received SSB signal at the UE is modeled as
\begin{equation}
    \mathbf y_{m,n} = \sqrt{P_{\mathrm{SSB}}}\, \mathbf h^H \mathbf c_{m,n} \mathbf s_{\mathrm{SSB}} + \mathbf n_{\mathrm{SSB},m,n},
    \label{eq:ssb_received_2d}
\end{equation}
where $P_{\mathrm{SSB}}$ is the SSB transmit power, $\mathbf s_{\mathrm{SSB}} \in \mathbb{C}^{L_s\times 1}$ collects the $L_s$ pilot symbols used for measurement, and $\mathbf n_{\mathrm{SSB},m,n}\sim \mathcal{CN}(\mathbf 0,\sigma^2\mathbf I_{L_s})$ denotes the corresponding noise vector.

The UE estimates the RSRP associated with $\mathbf c_{m,n}$ by averaging the received signal power over the $L_s$ SSB symbols, i.e., \cite{3gpp38215}
\begin{equation}
   r_{m,n} \triangleq \frac{1}{L_s}\sum_{l=1}^{L_s}|y_{m,n}[l]|^2,
    \label{eq:rsrp_def_2d}
\end{equation}
where $y_{m,n}[l]$ is the $l$-th entry of $\mathbf y_{m,n}$. Collecting the RSRP values over all probing beams yields the following two-dimensional RSRP fingerprint:
\begin{equation}
\mathbf R \triangleq [r_{m,n}] \in \mathbb{R}_{+}^{M_{\phi}\times M_{\theta}}.
\label{eq:rsrp_map}
\end{equation}
The UE then selects the access beam as
\begin{equation}
(\hat m,\hat n) = \arg\max_{m,n} r_{m,n},
\label{eq:best_beam_2d}
\end{equation}
and reports the corresponding beam index, or a shortlist of candidate beams, to the BS for subsequent access establishment and directional transmission.

Although the above procedure is simple and widely adopted, it becomes increasingly inefficient for large UPAs due to the following reasons \cite{giordani2019tutorial,song2017common}. 
\emph{First}, the sweeping overhead scales multiplicatively as $M_{\phi}M_{\theta}$, which grows rapidly when fine angular resolution is required in both dimensions. 
Moreover, reporting a two-dimensional beam fingerprint or multiple candidate beams also increases the feedback burden. 
\emph{Second}, the resulting beam is confined to a finite DFT codebook and thus suffers from beam quantization loss whenever the true dominant propagation direction does not align with the sampled codebook directions. 
Moreover, the same codebook is indiscriminately shared by all UEs and cannot explicitly exploit the site-specific propagation structure.

\subsection{Decoupled Channel Sensing}   \label{sec:sensing}

To alleviate the prohibitive overhead of exhaustive 2D sweeping, we consider a decoupled channel sensing strategy that probes the azimuth and elevation domains independently. Compared with the full UPA sweep, this strategy reduces the online sensing overhead from a multiplicative scaling to a linear one, while still preserving coarse directional information in both angular dimensions.

\subsubsection{Azimuth Sensing}

Let $\mathcal C_{\phi}=\{\mathbf c_1^{(\phi)},\ldots,\mathbf c_{M_{\phi}}^{(\phi)}\}$ denote an azimuth sensing codebook, where each codeword probes a different azimuth sector while maintaining broad coverage in elevation. A natural choice is
\begin{equation}
\mathbf c_m^{(\phi)} = \mathbf a_x(\phi_m,\theta_0)\kron \bar{\mathbf a}_y,
\end{equation}
where $\phi_m$ is a sampled azimuth angle, $\theta_0$ is a nominal elevation, and $\bar{\mathbf a}_y = \frac{1}{\sqrt{N_y}}\mathbf 1_{N_y}$ provides broad elevation coverage.
For $m$-th azimuth probing beam, the received SSB signal is
\begin{equation}
    \mathbf y_{m}^{(\phi)} = \sqrt{P_{\mathrm{SSB}}}\, \mathbf h^H \mathbf c_{m}^{(\phi)} \mathbf s_{\mathrm{SSB}} + \mathbf n_{\mathrm{SSB},m}^{(\phi)},
    \label{eq:ssb_received_az}
\end{equation}
where $\mathbf n_{\mathrm{SSB},m}^{(\phi)}\sim\mathcal{CN}(\mathbf 0,\sigma^2\mathbf I_{L_s})$. The corresponding azimuth-domain RSRP is estimated as
\begin{equation}
   r_{m}^{(\phi)} \triangleq \frac{1}{L_s}\sum_{l=1}^{L_s}|y_{m}^{(\phi)}[l]|^2.
    \label{eq:rsrp_def_az}
\end{equation}
The resulting azimuth-domain RSRP vector is
\begin{equation}
  \mathbf r_{\phi} = [r_{1}^{(\phi)}, \ldots, r_{M_{\phi}}^{(\phi)}]^T \in \mathbb{R}_{+}^{M_{\phi}\times 1}.
  \label{eq:obs_az}
\end{equation}

\subsubsection{Elevation Sensing}

Similarly, let $\mathcal C_{\theta}=\{\mathbf c_1^{(\theta)},\ldots,\mathbf c_{M_{\theta}}^{(\theta)}\}$ denote an elevation sensing codebook, where each codeword probes a different elevation sector while maintaining broad azimuth coverage. One representative construction is
\begin{equation}
\mathbf c_n^{(\theta)} = \bar{\mathbf a}_x \kron \mathbf a_y(\theta_n),
\end{equation}
where $\theta_n$ is a sampled elevation angle and $\bar{\mathbf a}_x = \frac{1}{\sqrt{N_x}}\mathbf 1_{N_x}$ provides broad azimuth coverage.
For the $n$-th elevation probing beam, the received SSB signal is given by
\begin{equation}
    \mathbf y_{n}^{(\theta)} = \sqrt{P_{\mathrm{SSB}}}\, \mathbf h^H \mathbf c_{n}^{(\theta)} \mathbf s_{\mathrm{SSB}} + \mathbf n_{\mathrm{SSB},n}^{(\theta)},
    \label{eq:ssb_received_el}
\end{equation}
where $\mathbf n_{\mathrm{SSB},n}^{(\theta)}\sim\mathcal{CN}(\mathbf 0,\sigma^2\mathbf I_{L_s})$. The corresponding elevation-domain RSRP is estimated as
\begin{equation}
   r_{n}^{(\theta)} \triangleq \frac{1}{L_s}\sum_{l=1}^{L_s}|y_{n}^{(\theta)}[l]|^2.
    \label{eq:rsrp_def_el}
\end{equation}
The resulting elevation-domain RSRP vector is
\begin{equation}
  \mathbf r_{\theta} = [r_{1}^{(\theta)}, \ldots, r_{M_{\theta}}^{(\theta)}]^T \in \mathbb{R}_{+}^{M_{\theta}\times 1}.
  \label{eq:obs_el}
\end{equation}

\subsubsection{Overhead Reduction and Information Loss}

A full 2D sweep preserves the joint azimuth-elevation beam fingerprint at the cost of $M_{\phi}M_{\theta}$ probing measurements \cite{noh2017multiresolution}. 
In contrast, separate sensing reduces the online overhead to $M_{\phi}+M_{\theta}$, which is much more scalable for large UPAs and short coherence intervals.
However, this overhead reduction comes at the expense of information loss, which propose more challenges for the beam generation model structure design. 

\begin{remark}
\normalfont
\emph{(Non-uniqueness of separate RSRP observations)}
The decoupled sensing vectors $\mathbf r_{\phi}$ and $\mathbf r_{\theta}$ only provide marginal power distributions in the azimuth and elevation domains, respectively, rather than the full joint angular structure of the UPA channel. As a result, multiple distinct channel realizations may induce highly similar observations $(\mathbf r_{\phi},\mathbf r_{\theta})$. For example, two channels with path directions $\{(\phi_1,\theta_1),(\phi_2,\theta_2)\}$ and $\{(\phi_1,\theta_2),(\phi_2,\theta_1)\}$ can exhibit nearly identical azimuth- and elevation-domain RSRP patterns, since they share the same azimuth and elevation marginals, while their pathwise azimuth--elevation pairings are different. Moreover, RSRP is a power-only measurement and thus discards the phase information of the multipath coefficients. Since the optimal transmit beamformer depends on the full steering vectors $\mathbf a(\phi_l,\theta_l)$ as well as the coherent combination of the complex path gains, such channels may require significantly different beamformers even when their separate RSRP observations are similar. Therefore, generate a high-fidelity beamformer from $(\mathbf r_{\phi},\mathbf r_{\theta})$ is fundamentally nontrivial, and requires additional structural priors, such as special designs in the beam generation model based on the site-specific propagation regularities.
\end{remark}


\subsection{Problem Formulation}

Our objective is to design a transmit beamforming vector that maximizes the received beamforming gain based only on the limited feedback obtained from decoupled sensing. Specifically, we seek to learn a beam generation function $f(\cdot)$ that maps the azimuth and elevation RSRP vectors, i.e., $(\mathbf r_{\phi},\mathbf r_{\theta})$, to a high-quality hardware-feasible beamformer. The problem is formulated as
\begin{subequations}
\begin{align}
\max_{f} \quad & \left|\mathbf{h}^H f\left(\mathbf r_{\phi}, \mathbf r_{\theta}\right)\right|^2 \notag \\
\text{s.t.} \quad & \left\| f\left(\mathbf r_{\phi}, \mathbf r_{\theta}\right)\right\|^2 = 1, \label{equ:w_norm} \\
& f\left(\mathbf r_{\phi}, \mathbf r_{\theta}\right) \in \mathcal W, \label{equ:w_feasible}
\end{align}
\end{subequations}
where $\mathcal W$ denotes the feasible set of analog beamformers determined by the underlying hardware architecture, most notably the constant-modulus constraint on each antenna element.

The above problem is challenging for several reasons. 
\emph{First,} the observations $\mathbf r_{\phi}$ and $\mathbf r_{\theta}$ are nonlinear noisy power measurements, rather than linear channel observations, and therefore do not contain explicit phase information. \emph{Second,} decoupled sensing provides only marginal angular information and does not preserve the joint azimuth--elevation coupling of the multipath channel. \emph{Third,} the desired beamformer lies in the nonconvex feasible set $\mathcal W$, which prevents direct application of standard convex optimization tools. As a result, the mapping from $(\mathbf r_{\phi},\mathbf r_{\theta})$ to the near-optimal beamformer is highly nonconvex and generally non-unique, making direct deterministic model-based inversion intractable.

The above observations suggest that efficient beamforming generation should exploit additional structural priors beyond instantaneous online measurements. In particular, by leveraging site-specific propagation regularities learned offline, design the special generative model for the beamform generation.
Making it becomes more possible to infer a high-fidelity beamformer from low-dimensional RSRP feedback with significantly reduced online overhead.

\section{Cross-Fused GenSSBF Framework}
\label{sec:framework}

In this section, we present the proposed cross-fused GenSSBF framework. 
We first introduce the angular-domain channel encoding representation adopted for training in Section \ref{subsec:angular_encoding}.
Then, we describe the cross-fused encoder that mitigates the information loss of separate sensing in Section \ref{subsec:cross_fusion}.
Next, we develop the conditional normalizing-flow model as the beam generator in Section~\ref{subsec:flow_generation}.

\subsection{Angular-Domain Channel Encoding}
\label{subsec:angular_encoding}

We leverage the sparse nature of far-field channels in the angular domain to perform beam generation on the DFT representation $\tilde{\mathbf{h}}=\mathcal{F}(\mathbf{h})$, where $\mathcal{F}(\cdot)$ denotes the DFT operator.
Let the $i$-th angular-domain coefficient be written as
\begin{equation}
\tilde{\mathbf h}[i] = \rho_i e^{j\varphi_i}, \qquad i=1,\ldots,N_t,
\label{eq:h_representation}
\end{equation}
where $\rho_i = |\tilde{\mathbf h}[i]|$ and $\varphi_i = \angle \tilde{\mathbf h}[i]$. Instead of using the phase $\varphi_i$ directly, we encode each coefficient by its amplitude together with the cosine-sine pair of its phase, and form the real-valued training data as follows:
\begin{equation}
\mathbf x
=
\big[
\boldsymbol{\rho}^T,\;
\cos(\boldsymbol{\varphi})^T,\;
\sin(\boldsymbol{\varphi})^T
\big]^T
\in \mathbb{R}^{3N_t},
\label{eq:beamspace_encoding}
\end{equation}
where $\boldsymbol{\rho}=[\rho_1,\ldots,\rho_{N_t}]^T$ and $\boldsymbol{\varphi}=[\varphi_1,\ldots,\varphi_{N_t}]^T$. 

The cosine-sine phase encoding is adopted to avoid the phase-wrapping issue of direct angle regression. If the phase is represented only by $\varphi_i\in[-\pi,\pi)$, two physically close angles near the branch cut, e.g., $\pi-\epsilon$ and $-\pi+\epsilon$, would appear numerically far apart, thereby introducing an artificial discontinuity into the learning target. By contrast, the pair $(\cos\varphi_i,\sin\varphi_i)$ lies on the unit circle and varies smoothly with the phase, which leads to a smoother loss landscape and more stable gradient propagation. 
In addition, the cosine-sine representation is bounded in $[-1,1]$, which further improves numerical stability. The original phase can be uniquely recovered as $\varphi_i = \mathrm{atan2}\big(\sin\varphi_i,\cos\varphi_i\big)$, and thus the phase information is preserved without ambiguity.

\subsection{Cross-Fused Encoder}
\label{subsec:cross_fusion}
 
Decoupled channel sensing reduces the online probing overhead, but it also destroys the explicit joint coupling between azimuth and elevation. 
To alleviate this information loss, we introduce a lightweight cross-fused condition encoder, whose role is to extract branch-specific features from the two one-dimensional RSRP observations and then exchange complementary information between them before forming the conditioning vector for the downstream generative beamformer.

\subsubsection{Independent Branch Embedding}

The azimuth and elevation domain RSRP vectors are first processed by two separate MLP encoders:
\begin{align}
\mathbf e_\phi = \mathcal E_\phi(\mathbf r_\phi) \in \mathbb R^{M_{\phi}\times d_e}, \label{eq:embed_azimuth}\\
\mathbf e_\theta = \mathcal E_\theta(\mathbf r_\theta) \in \mathbb R^{M_{\theta}\times d_e}, \label{eq:embed_theta}
\end{align}
where $\mathcal E_\phi(\cdot)$ and $\mathcal E_\theta(\cdot)$ are lightweight trainable encoders. 
Each row of $\mathbf e_\phi$ is a $d$-dimensional feature vector associated with one azimuth sensing beam, and similarly for $\mathbf e_\theta$. These feature matrices serve as the token sequences for the subsequent cross-attention operation.

\subsubsection{Bidirectional Cross-Attention}

To exploit the mutual dependence between the azimuth and elevation sensing branches, we adopt a bidirectional cross-attention block. The core idea is to let one branch use its own feature as a \emph{query} to retrieve complementary information from the other branch, whose feature serves as the corresponding \emph{key} and \emph{value} \cite{vaswani2017attention, tan2019lxmert}. 
In this way, the encoder can learn the latent azimuth-elevation coupling that is not explicitly available from separate sensing.

\paragraph*{Azimuth conditioned on elevation}
To refine the azimuth feature matrix using elevation information, we first project $\mathbf e_\phi$ and $\mathbf e_\theta$ into query, key, and value spaces:
\begin{align}
\mathbf{Q}_{\phi} &= \mathbf{e}_{\phi}\mathbf{W}_{Q}^{\phi}
\in \mathbb{R}^{M_{\phi}\times d_k},\\
\mathbf{K}_{\theta} &= \mathbf{e}_{\theta}\mathbf{W}_{K}^{\theta}
\in \mathbb{R}^{M_{\theta}\times d_k},\\
\mathbf{V}_{\theta} &= \mathbf{e}_{\theta}\mathbf{W}_{V}^{\theta}
\in \mathbb{R}^{M_{\theta}\times d_v}.
\end{align}
%
where $\mathbf W_Q^\phi \in \mathbb R^{d_e\times d_k}$, $\mathbf W_K^\theta \in \mathbb R^{d_e\times d_k}$, and $\mathbf W_V^\theta \in \mathbb R^{d_e\times d_v}$ are trainable projection matrices. Here, the query matrix $\mathbf Q_\phi$ specifies what information each azimuth beam feature seeks, the key matrix $\mathbf K_\theta$ determines how compatible each elevation feature is with that request, and the value matrix $\mathbf V_\theta$ contains the elevation-side information to be aggregated.
The corresponding attention weights are computed as
\begin{equation}
\mathbf A_{\phi \leftarrow \theta}
=
\softmax\!\left(
\frac{\mathbf Q_\phi \mathbf K_\theta^T}{\sqrt{d_k}}
\right)
\in \mathbb R^{M_{\phi}\times M_{\theta}},
\label{eq:attn_phi_theta}
\end{equation}
where the softmax is applied row-wise. The $(m,n)$-th entry of $\mathbf A_{\phi \leftarrow \theta}$ measures the relevance of the $n$-th elevation feature to the $m$-th azimuth feature. The refined azimuth representation is then obtained by
\begin{equation}
\tilde{\mathbf e}_\phi
=
\mathbf A_{\phi \leftarrow \theta}\mathbf V_\theta
\in \mathbb R^{M_{\phi}\times d_v}.
\label{eq:cross_phi}
\end{equation}
Equivalently, the $m$-th refined azimuth feature is a weighted combination of all elevation-side value vectors, with weights determined by their compatibility with the $m$-th azimuth query.

\paragraph*{Elevation conditioned on azimuth}
Symmetrically, the elevation branch is refined by using its own features as queries and the azimuth features as keys and values:
\begin{align}
\mathbf Q_\theta = \mathbf e_\theta \mathbf W_Q^\theta \in \mathbb R^{M_{\theta}\times d_k},\\
\mathbf K_\phi = \mathbf e_\phi \mathbf W_K^\phi \in \mathbb R^{M_{\phi}\times d_k}, \\
\mathbf V_\phi = \mathbf e_\phi \mathbf W_V^\phi \in \mathbb R^{M_{\phi}\times d_v},
\end{align}
where $\mathbf W_Q^\theta \in \mathbb R^{d_e\times d_k}$, $\mathbf W_K^\phi \in \mathbb R^{d_e\times d_k}$, and $\mathbf W_V^\phi \in \mathbb R^{d_e\times d_v}$ are trainable matrices. The attention weights and refined elevation representation are given by
\begin{equation}
\mathbf A_{\theta \leftarrow \phi}
=
\softmax\!\left(
\frac{\mathbf Q_\theta \mathbf K_\phi^T}{\sqrt{d_k}}
\right)
\in \mathbb R^{M_{\theta}\times M_{\phi}},
\label{eq:attn_theta_phi}
\end{equation}
and
\begin{equation}
\tilde{\mathbf e}_\theta
=
\mathbf A_{\theta \leftarrow \phi}\mathbf V_\phi
\in \mathbb R^{M_{\theta}\times d_v}.
\label{eq:cross_theta}
\end{equation}


After the cross-attention update, the features are further refined using residual connections, layer normalization, and branch-wise feed-forward networks:
\begin{align}
\bar{\mathbf e}_\phi &= \mathrm{LN}\!\left(\mathbf e_\phi + \tilde{\mathbf e}_\phi\right), \\
\bar{\mathbf e}_\theta &= \mathrm{LN}\!\left(\mathbf e_\theta + \tilde{\mathbf e}_\theta\right), \\
\hat{\mathbf e}_\phi &= \mathrm{LN}\!\left(\bar{\mathbf e}_\phi + \mathrm{FFN}_\phi(\bar{\mathbf e}_\phi)\right), \\
\hat{\mathbf e}_\theta &= \mathrm{LN}\!\left(\bar{\mathbf e}_\theta + \mathrm{FFN}_\theta(\bar{\mathbf e}_\theta)\right),
\label{eq:post_attn}
\end{align}
where $\mathrm{FFN}_\phi(\cdot)$ and $\mathrm{FFN}_\theta(\cdot)$ denote two lightweight branch-wise feed-forward networks.

\subsubsection{Condition Vector Formation}

The refined branch embeddings are concatenated and mapped to a condition vector:
\begin{equation}
\mathbf c
=
\mathcal E_c\!\left(
[\hat{\mathbf e}_\phi^T,\hat{\mathbf e}_\theta^T]^T
\right)
\in \mathbb R^{d_c},
\label{eq:latent}
\end{equation}
where $\mathcal E_c(\cdot)$ is a fusion MLP and $d_c$ is the condition dimension. The resulting vector $\mathbf c$ summarizes the site-specific two-dimensional angular structure inferred from the separate sensing observations. It also serves as the conditioning input to the subsequent normalizing-flow beamformer generator.

\subsection{Conditional Normalizing Flow Model}
\label{subsec:flow_generation}

Due to the ambiguity of decoupled sensing, the mapping from the fused condition vector $\mathbf c$ to the final beamformer is generally one-to-many. 
Therefore, a deterministic predictor tends to produce a compromised solution when multiple beam directions are possible under the same observation. To address this issue, we adopt a \emph{conditional normalizing flow} as the beam generation module. The flow serves as a stochastic generator that maps a simple latent Gaussian variable to a complex conditional distribution over beam-space representations, thereby enabling the generation of multiple high-quality beam candidates from the same sensing observation.

Specifically, let $\mathbf x \in \mathbb R^{3N_t}$ denote the real-valued beam-space representation introduced in \eqref{eq:beamspace_encoding}, and let $\mathbf z \in \mathbb R^{3N_t}$ be a latent random vector drawn from the standard Gaussian prior $\mathcal N(\mathbf 0,\mathbf I_{3N_t})$. We employ a conditional normalized flow model with the following forward transform \cite{real-nvp}:
\begin{equation}
\mathbf z = \Psi_{\boldsymbol{\xi}}(\mathbf x;\mathbf c).
\end{equation}
The corresponding inverse transform is given by
\begin{equation}
\mathbf x = G_{\boldsymbol{\xi}}(\mathbf z,\mathbf c) = \Psi_{\boldsymbol{\xi}}^{-1}(\mathbf z;\mathbf c),
\label{eq:flow_inverse_gen}
\end{equation}
where $\boldsymbol{\xi}$ collects the trainable flow parameters. The inverse transform in \eqref{eq:flow_inverse_gen} is used for beam generation, while the forward transform provides an exact density model
\begin{equation}
\log p_{\boldsymbol{\xi}}(\mathbf x \mid \mathbf c)
=
\log p(\mathbf z)
+
\sum_{\ell=1}^{L_{\mathrm{flow}}}
\log
\left|
\det
\frac{\partial \Psi_{\ell}}{\partial \mathbf u^{(\ell-1)}}
\right|,
\label{eq:flow_logprob}
\end{equation}
where $L_{\mathrm{flow}}$ is the number of coupling layers, $\mathbf u^{(0)}=\mathbf x$, $\mathbf u^{(L_{\mathrm{flow}})}=\mathbf z$, and $\Psi_{\ell}$ denotes the $\ell$-th invertible transformation. Although the model admits exact log-likelihood evaluation, in the proposed framework this property is used primarily to ensure stable invertible generation rather than to define the final training objective.

\begin{figure}[t]
\begin{center}
\setlength{\abovecaptionskip}{0cm}
\includegraphics[width=6cm, trim = 0cm 0cm 0cm 0cm, clip = true]{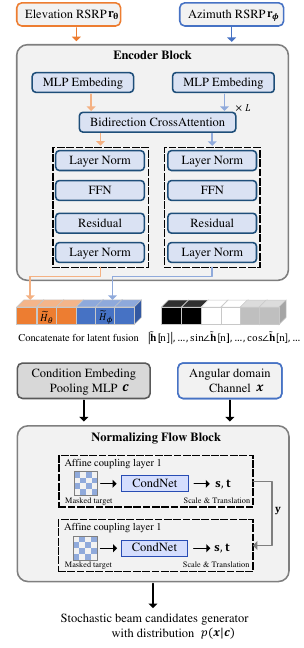}
\caption{\footnotesize{Architecture of the proposed cross-fused GenSSBF framework.}}
\label{fig:model_arc}
\end{center}
\end{figure}

\subsubsection{Affine Coupling Layers}

Each flow layer is implemented by a conditional affine coupling transformation. 
For the $\ell$-th layer, let $\mathbf m^{(\ell)} \in \{0,1\}^{3N_t}$ denote a binary mask, and define the two complementary subvectors
\begin{align}
\mathbf x_A^{(\ell)} = \mathbf m^{(\ell)} \odot \mathbf x^{(\ell-1)}, \\
\mathbf x_B^{(\ell)} = \left(\mathbf 1-\mathbf m^{(\ell)}\right)\odot \mathbf x^{(\ell-1)},
\end{align}
where $\odot$ denotes the Hadamard product. The masked part $\mathbf x_A^{(\ell)}$ is kept unchanged, while the complementary part $\mathbf x_B^{(\ell)}$ is updated through an affine transformation conditioned on both $\mathbf x_A^{(\ell)}$ and the condition sensing feature $\mathbf c$:
\begin{align}
\mathbf y_A^{(\ell)} &= \mathbf x_A^{(\ell)}, \\
\mathbf y_B^{(\ell)} &= \mathbf x_B^{(\ell)} \odot \exp\!\big(\mathbf s^{(\ell)}\big) + \mathbf t^{(\ell)},
\label{eq:affine_coupling_gen}
\end{align}
where
\begin{equation}
[\mathbf s^{(\ell)},\mathbf t^{(\ell)}]
=
\mathrm{CondNet}^{(\ell)}
\!\left(
\left[
\mathbf x_A^{(\ell)}[\mathcal I_A^{(\ell)}]^T,\;
\mathbf c^T
\right]^T
\right).
\label{eq:condnet_gen}
\end{equation}
Here, $\mathcal I_A^{(\ell)}$ denotes the active index set selected by the mask, and $\mathrm{CondNet}^{(\ell)}(\cdot)$ can be a lightweight MLP. More particularly, the unchanged coordinates provide a partial observation of the current sample, while the condition vector $\mathbf c$ injects the site-specific sensing information. Then, their concatenation is used to predict the scale and translation parameters for the remaining coordinates.
%
%
Thus, the inverse transformation is
\begin{align}
\mathbf x_A^{(\ell)} &= \mathbf y_A^{(\ell)}, \\
\mathbf x_B^{(\ell)} &= \left(\mathbf y_B^{(\ell)} - \mathbf t^{(\ell)}\right)\odot \exp\!\big(-\mathbf s^{(\ell)}\big),
\end{align}
which guarantees exact sample generation and efficient backpropagation through the reparameterized sampling path.

\subsubsection{Jacobian Determinant}

A major advantage of affine coupling layers is that the Jacobian matrix is triangular. Therefore, its determinant can be computed exactly without explicit matrix inversion. For the $\ell$-th layer, we have
\begin{equation}
\det
\frac{\partial \mathbf y^{(\ell)}}{\partial \mathbf x^{(\ell-1)}}
=
\prod_{j \in \mathcal I_B^{(\ell)}} \exp\!\left(s_j^{(\ell)}\right),
\end{equation}
which leads to
\begin{equation}
\log
\left|
\det
\frac{\partial \mathbf y^{(\ell)}}{\partial \mathbf x^{(\ell-1)}}
\right|
=
\sum_{j \in \mathcal I_B^{(\ell)}} s_j^{(\ell)}.
\label{eq:logdet_gen}
\end{equation}
This exact log-determinant computation is one of the key reasons for adopting the normalized flow model architecture.

\subsubsection{Scale Clamping}

To further improve numerical stability, the scale output is clamped by
\begin{equation}
\mathbf s^{(\ell)} = 2\tanh\!\left(\mathbf s_{\mathrm{raw}}^{(\ell)}\right),
\label{eq:scale_clamp_gen}
\end{equation}
which ensures $s_j^{(\ell)} \in [-2,2]$ for all updated coordinates. Consequently, the multiplicative factor $\exp(s_j^{(\ell)})$ is bounded within $[e^{-2},e^2]$, preventing unstable amplification or attenuation across successive coupling layers.

\subsubsection{Alternating Mask Schedule}

To ensure that every coordinate of $\mathbf x$ is updated by the flow, we adopt an alternating binary mask schedule across layers. 
For the $\ell$-th coupling layer, the mask is defined over the flattened vector coordinates as
\begin{equation}
\!\!
\text{mask}^{(\ell)}[n']
\!=\!
\begin{cases}
1, \!\!\!\!  & \mathrm{mod}(n'\!+\!\ell,2)=0,\\
0, \!\!\!\!  & \text{otherwise},
\end{cases}
n'=1,\ldots,3N_t.
\label{eq:alt_mask}
\end{equation}
Therefore, the coordinates that remain unchanged in one layer are updated in the next layer. After a small number of layers, all entries of the beam-space representation have interacted through the nonlinear conditional transformations, enabling effective mixing across amplitude, cosine-phase, and sine-phase components.

With the above design, conditioned on the fused sensing feature $\mathbf c$, the BS generates multiple latent vectors $\{\mathbf z_k\}$ from the standard Gaussian prior, and maps them through \eqref{eq:flow_inverse_gen} to obtain the corresponding beam-space candidates $\{\mathbf x_k\}$, which are subsequently decoded into feasible transmit beamformers.

\section{Task-Oriented Training and Online Beam Prediction}
\label{sec:training_infer}

In this section, we present the training and deployment procedures of the proposed GenSSBF framework. 
We first introduce the normalized gain loss and explain its connection to the generative design in Section \ref{subsec:gain_training}. 
Then, we present the candidate beam selection mechanism in Section \ref{subsec:online_selection}.
Finally, we summarize the complete offline training and online inference procedures in Section \ref{subsec:algorithm}.

\subsection{Task-Oriented Training}
\label{subsec:gain_training}

Let $\mathcal D=\{(\mathbf r_{\phi,d},\mathbf r_{\theta,d},\mathbf x_d)\}_{d=1}^{D}$ denote the site-specific training set, where $\mathbf x_d$ is the real-valued representation of  $\mathbf h_d \in \mathbb C^{N_t\times 1}$, which is the downlink channel corresponding to the $d$-th user location, and $(\mathbf r_{\phi,d},\mathbf r_{\theta,d})$ are the associated separate azimuth and elevation domain RSRP observations. For a mini-batch of size $B$, the cross-fused encoder first maps the RSRP to the condition vector $\mathbf c_b$, $b=1,\ldots,B$. Based on $\mathbf c_b$, the flow generator produces $K$ beam-space candidates according to
\begin{equation}
\hat{\mathbf x}_{b,k} = G_{\boldsymbol{\xi}}(\mathbf z_{b,k},\mathbf c_b),
\qquad
\mathbf z_{b,k}\sim \mathcal N(\mathbf 0,\mathbf I_{3N_t}),
\label{eq:train_sample_gen}
\end{equation}
where $\mathbf z_{b,k}$ is the Gaussian vector. Each generated $\hat{\mathbf x}_{b,k}$ is then decoded as the spatial-domain beam vector $\hat{\mathbf w}_{b,k}\in\mathcal W$, satisfying $\|\hat{\mathbf w}_{b,k}\|^2=1$.
The decoding procedure is from the angular domain to the real-complex domain, based on \eqref{eq:h_representation} and \eqref{eq:beamspace_encoding}.
The training objective is to maximize the best gain among the $K$ generated candidates.
Therefore, the loss function can be defined as
\begin{equation}
\mathcal L_{\rm gain}
=
1
-
\frac{1}{B}
\sum_{b=1}^{B}
\max_{1\le k\le K}
g\!\left(\hat{\mathbf w}_{b,k},\mathbf h_b\right).
\label{eq:gain_loss}
\end{equation}
where 
\begin{equation}
g(\hat{\mathbf w}_{b,k},\mathbf h_b) \triangleq \frac{|\mathbf h_b^H \hat{\mathbf w}_{b,k}|^2} {\|\mathbf h_b\|^2}
\label{eq:gain}
\end{equation}
is the normalized beamforming gain.

The role of $K$ is to control the candidate diversity during training. When $K=1$, the model generates only one beam candidate for each condition, and the parameter update is determined by the gain of this single random sample. In this case, the training gradient may have high variance, since an occasionally poor sample can lead to an overly pessimistic loss even if the conditional generator has learned useful high-gain modes. When $K>1$, multiple candidates are sampled from the similar conditional distribution, and only the candidate with the largest beamforming gain contributes to the loss. Therefore, the model is encouraged to assign probability mass to high-quality beamforming modes, rather than merely producing an average or compromised beamformer. This best-of-$K$ design is particularly suitable for decoupled channel sensing, where the similar RSRP observations may correspond to multiple beam directions.

\begin{algorithm}[t]
\caption{Offline Training of Cross-fused GenSSBF}
\label{alg:training}
\textbf{Input:} Training dataset $\mathcal D=\{(\mathbf r_{\phi,d},\mathbf r_{\theta,d},\mathbf x_d)\}_{d=1}^{D}$, batch size $B$, learning rate $\eta$, number of epochs $I$, generation times $K$ \\
\textbf{Output:} The trained model $\Phi^\star$ comprises cross-fused encoder and conditional flow generator
\begin{algorithmic}[1]
\STATE Model parameter initialization $\Phi^{(0)}$.
\FOR{$i = 1,2,\ldots,I$}
        \STATE Sample a mini-batch of data $\mathcal{B}_i \subset \mathcal{D}$ with $|\mathcal{B}_i| = B$.
        \STATE Compute condition vectors $\{\mathbf c_b\}_{b=1}^{B}$ from the separate RSRP observations $(\mathbf r_{\phi,d},\mathbf r_{\theta,d})\in \mathcal B_i$ by \eqref{eq:embed_azimuth}-\eqref{eq:latent}.
        \FOR{$k = 1,2,\ldots,K$}
            \STATE Define latent Gaussian vectors $\{\mathbf z_{b,k}\}_{b=1}^{B}$.
            \STATE Generate beam-space candidates $\hat{\mathbf x}_{b,k}$ by \eqref{eq:train_sample_gen}.
            \STATE Decode $\hat{\mathbf x}_{b,k}$ to obtain feasible beamformers $\hat{\mathbf w}_{b,k}$.
            \STATE Evaluate normalized beamforming gains $g(\hat{\mathbf w}_{b,k},\mathbf h_b)$.
        \ENDFOR
        \STATE Compute the mini-batch loss $\mathcal L_{\rm gain}$ by \eqref{eq:gain_loss}.
        \STATE Update $\Phi$ by
        $
        \Phi^{(i)} \leftarrow \Phi^{(i-1)} - \eta^{(i)} \nabla_{\Phi}\mathcal L_{\rm gain}.
        $
\ENDFOR
\STATE The trained model is obtained as $\Phi^\star=\Phi^{(I)}$.
\end{algorithmic}
\end{algorithm}

\textit{Training Objective:} 
The proposed framework can be trained using either the negative normalized beamforming gain in \eqref{eq:gain_loss} or the negative log-likelihood (NLL) induced by \eqref{eq:flow_logprob}. 
In this work, we adopt the gain-based loss due to two reasons.
1) NLL penalizes discrepancies over the entire beam-space target representation. Hence, representation-level perturbations, such as small amplitude errors or phase deviations in non-dominant components, may incur a noticeable likelihood penalty. While this may have only limited impact on the final beamforming gain. 
2) Decoupled sensing generally induces a multimodal conditional mapping from RSRP to the beamformer. As a result, multiple distinct beamformers may be consistent with the similar sensing observation. 
In such a case, the practical objective depends mainly on whether the generated candidate set covers at least one high-gain mode. 
This motivates the following high-gain coverage interpretation.

\begin{proposition} \label{prop:high-gain}
\normalfont
\emph{(High-gain coverage under finite candidate generation)}
For fixed condition $\mathbf c$ and channel $\mathbf h$, let $q_{\boldsymbol{\xi}}(\hat{\mathbf w}\mid \mathbf c)$ denote the beamformer distribution induced by the conditional generator. For any $\epsilon\in(0,1)$, we define the high-gain set as
\begin{equation}
\mathcal A_{\epsilon}(\mathbf c,\mathbf h)
\triangleq
\left\{
\hat{\mathbf w}\in\mathcal W:
g(\hat{\mathbf w},\mathbf h)\ge 1-\epsilon
\right\}.
\label{eq:high_gain_set}
\end{equation}
If $K$ beamformers are independently generated from $q_{\boldsymbol{\xi}}(\hat{\mathbf w}\mid \mathbf c)$, the probability that at least one candidate belongs to $\mathcal A_{\epsilon}(\mathbf c,\mathbf h)$ is
\begin{equation}
P_{\mathrm{succ}}^{(K)}(\mathbf c,\mathbf h;\epsilon)
=
1-
\left(
1-
q_{\boldsymbol{\xi}}\!\left(
\mathcal A_{\epsilon}(\mathbf c,\mathbf h)
\,\middle|\,
\mathbf c
\right)
\right)^K.
\label{eq:success_probability}
\end{equation}
\end{proposition}

\begin{proof}
Since the $K$ candidates are independently generated, \eqref{eq:success_probability} follows directly from the complement probability that none of them belongs to $\mathcal A_{\epsilon}(\mathbf c,\mathbf h)$.
\end{proof}

\textbf{Proposition} \ref{prop:high-gain} shows that the key quantity is the probability mass assigned to the high-gain region. 
In other words, successful beam prediction depends primarily on whether the generator can produce beamformers that are effective for the final beam alignment task.
The normalizing flow model works as a stochastic candidate generator rather than as a likelihood estimator.
During training, the best candidate is identified using the ground-truth channel.
Whereas during deployment, it is selected through short pilot-based measurements. 


\begin{algorithm}[t]
\caption{Online Beam Prediction and Selection}
\label{alg:inference}
\textbf{Input:} Trained cross-fused GenSSBF model $\Phi^\star$\\
\textbf{Output:} Final beamformer $\mathbf w^\star$
\begin{algorithmic}[1]
\STATE BS sweeps the azimuth sensing codebook $\mathcal C_{\phi}$ and the elevation sensing codebook $\mathcal C_{\theta}$ with SSB.
\STATE UE acquires the separate RSRP observations $(\mathbf r_\phi,\mathbf r_\theta)$.
\STATE \textit{$\triangleright$ Based on the trained model $\Phi^\star$:}
\STATE \hspace{1em} Compute condition vector $\mathbf{c}$ from $(\mathbf{r}_\phi, \mathbf{r}_\theta)$.
\STATE \hspace{1em} Generate the set of latent Gaussian vectors $\{\mathbf{z}_k\}$.
\STATE \hspace{1em} Obtain the beam-space candidate $\hat{\mathbf x}_k=G_{\boldsymbol{\xi}}(\mathbf z_k,\mathbf c)$.
    \STATE Decode $\{\hat{\mathbf x}_k\}$ to obtain the candidate beamformer $\{\hat{\mathbf w}_k\}$.
\STATE UE feeds back the measured RSRP vector $\{r_k\}$.
\STATE BS applies the final beamforming vector $\mathbf w^\star$ by \eqref{eq:best_candidate_k}.
\end{algorithmic}
\end{algorithm}

\subsection{Online Candidate Verification and Beam Selection}
\label{subsec:online_selection}

For the $k$-th candidate beamformer $\hat{\mathbf w}_k$, the received SSB signal of the UE is given by
\begin{equation}
    \mathbf y_{k} = \sqrt{P_{\mathrm{SSB}}}\, \mathbf h^H \hat{\mathbf w}_k \mathbf s_{\mathrm{SSB}} + \mathbf n_{\mathrm{SSB}},
    \label{eq:ssb_received_el}
\end{equation}
The corresponding RSRP can be represented as
\begin{equation}
   r_{k} \triangleq \frac{1}{L_s}\sum_{l=1}^{L_s}|y_{k}[l]|^2.
    \label{eq:rsrp_def_el}
\end{equation}
Based on these measurements, the UE feeds back only the index of the best candidate
\begin{equation}
k^\star = \argmax_{1\le k\le K}r_k.
\label{eq:best_candidate_k}
\end{equation}
The final beam used for downlink transmission is $\mathbf w^\star = \hat{\mathbf w}_{k^\star}$.

\subsection{Offline Training and Online Inference Procedures}
\label{subsec:algorithm}

The complete offline training and online inference procedures of the proposed cross-fused GenSSBF framework are summarized in \textbf{Algorithms}~\ref{alg:training} and \ref{alg:inference}, respectively. 
For ease of representation, we define $\Phi^\star$ as the trained model, which comprises both the cross-fused condition encoder and the conditional normalized flow generator.

In the offline training stage, the site-specific dataset $\mathcal D=\{(\mathbf r_{\phi,d},\mathbf r_{\theta,d},\mathbf x_d)\}_{d=1}^{D}$ with $(\mathbf r_{\phi,d},\mathbf r_{\theta,d})$ as the separate azimuth and elevation domain RSRP observations, and $\mathbf x_d$ as the corresponding supervision sample used for beam generation. During each training iteration $i$, a mini-batch data $\mathcal{B}_i$ is first drawn from $\mathcal D$, and the separate RSRP observations are mapped to condition vectors through the cross-fused encoder. Conditioned on these features, the flow model generates $K$ beam-space candidates from independent latent Gaussian vectors. 
Then, the generated candidates are decoded into feasible beamformers, and their normalized beamforming gains are evaluated. 
Finally, the model parameters are updated by minimizing the task-oriented loss.
It enables the model to generate a compact candidate set that contains at least one high-gain beamformer from the individual RSRP observations.

In the online inference stage, the BS first performs separate beam sweeping over the azimuth sensing codebook $\mathcal C_{\phi}$ and the elevation sensing codebook $\mathcal C_{\theta}$ through SSB transmission. Based on the received pilots, the UE acquires the separate RSRP observations $(\mathbf r_\phi,\mathbf r_\theta)$ and feeds them to the trained model $\Phi^\star$. 
Then, the cross-fused encoder then computes the condition vector $\mathbf c$, after which the conditional flow generator produces a set of stochastic beam-space candidates from sampled latent Gaussian vectors. These candidates are decoded into beamforming vectors. The UE reports the measured candidate RSRP values, and the BS selects the final beamformer according to \eqref{eq:best_candidate_k}. 
In this way, the proposed cross-fused GenSSBF model achieves robust beam prediction under ambiguous separate sensing RSRP observations.

\section{Simulation Results}
\label{sec:simulation}

In this section, we present the numerical results to evaluate the performance of our proposed cross-fused GenSSBF framework. 
We first describe the basic simulation settings in Section \ref{subsec:simulation_setting}.
Then, we discuss the baseline schemes in Section \ref{subsec:baseline}. 
Finally, we demonstrate and compare the generation performance of the cross-fused GenSSBF scheme with the baseline schemes in Section \ref{subsec:performance}.

\begin{table}[t]
\caption{System Parameters}
\label{tab:system parameters}
\centering
\begin{tabular}{ll}
\hline
{\textbf{Parameter}} & {\textbf{Value}}\\
\hline
\hline
{{BS antenna, $N_x\times N_y$}} & {$16\times 16$ UPA}\\
{{BS full  DFT codebook size}} & {$32\times 32$}\\
{Antenna element} & Isotropic\\
{Antenna spacing, $d$} & $\lambda/2$\\
{Carrier frequency, $f_c$} & $28$ GHz\\
{Number of paths, $L$} & $5$\\ 
{Transmit power $P_t$} & $40$ dBm\\
{Candidate beam number in training $K$} & $8$\\
{Number of SSB symbols, $L_s$} & $5$\\
{Training epoch, $I$} & $200$\\
{Cosine annealing learning rate, $\eta_{\max}$} & $2\times 10^{-3}$\\
{RSRP embedding dimension, $d_e$} & $32$\\
{Cross-attention layers, heads} & $2$, $4$\\ 
{Condition embedding dimension, $d_c$} & $64$\\
{Affine coupling layers, $L_{\mathrm{flow}}$} & $2$\\
{Hidden dim in CondNet} & $16$\\
\hline
\end{tabular}
\end{table}
%

\subsection{Simulation Settings}
\label{subsec:simulation_setting}

\subsubsection{Downlink beam prediction setup}
We consider a narrowband downlink communication system, where the BS is equipped with a $16\times 16$ UPA and each UE is equipped with a single antenna. For controlled validation, we first adopt a sparse geometric channel model with $L=5$ dominant propagation paths per UE. The azimuth and elevation angles are independently generated according to $\phi \sim \mathcal{U}(-\pi,\pi)$ and $\theta \sim \mathcal{U}(-\pi/2,\pi/2)$, respectively.

\subsubsection{Evaluation scenarios}
To evaluate the proposed framework under realistic site-specific environments, we use the DeepMIMO dataset \cite{deepmimo}.
It provides ray-tracing-based channel realizations for wireless learning and beam management studies.\footnote{Scenario description is available at \url{https://www.deepmimo.net/scenarios}.} Specifically, we consider the I2\_28, O1B\_28, and Boston5G\_28 scenarios. The I2\_28 scenario represents a 28 GHz indoor environment without LoS propagation. The O1B\_28 scenario corresponds to a more challenging outdoor environment with stronger NLoS effects. The Boston5G\_28 scenario considers an urban environment with richer multipath propagation and more complex blockage conditions. These scenarios jointly cover LoS-dominant, blockage-limited, and dense urban propagation environments in the mmWave band.

\subsubsection{Learning model}
We adopt a conditional Real-NVP model as the beamformer generator in the proposed framework, with a total of 185,024 trainable parameters \cite{real-nvp}. Each parameter is stored as a 32-bit floating-point value. The model is trained using a cosine-annealing learning-rate schedule with initial learning rate $\eta_{\max}=2\times 10^{-3}$ and minimum learning rate $\eta_{\min}=1\times 10^{-5}$. For each scenario, $80\%$ of the channel realizations are used for training and the remaining $20\%$ are used for testing. The detailed system parameters and training hyperparameters are summarized in Table~\ref{tab:system parameters}.

\begin{figure}[t]
\centering
\setlength{\abovecaptionskip}{0cm}
\includegraphics[width=8.5cm, trim = 0cm 0cm 0cm 0cm, clip = true]{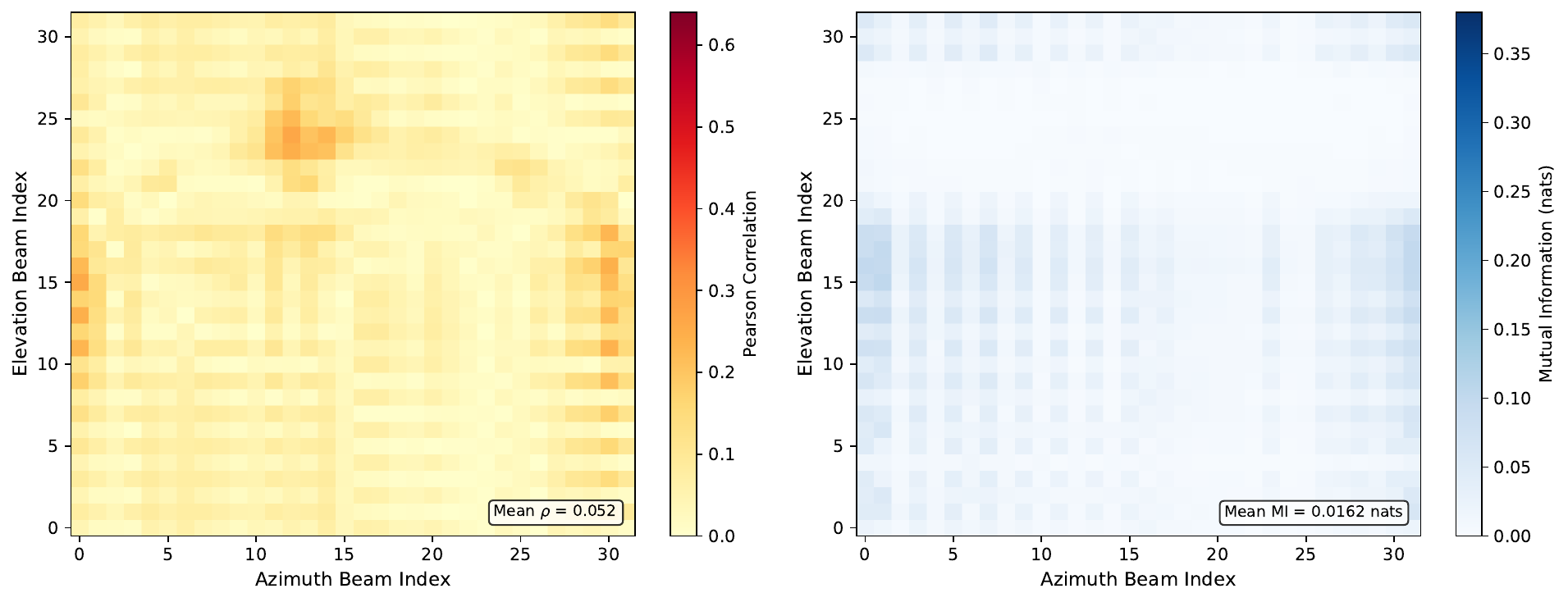}
\caption{\footnotesize{Cross-correlation between the separate azimuth and elevation domain RSRP observations before cross-attention in the O1B\_28 scenario.}}
\label{fig:rsrp_corr_before_o1b28}
\end{figure}

\begin{figure}[t]
\centering
\setlength{\abovecaptionskip}{0cm}
\includegraphics[width=8.5cm, trim = 0cm 0cm 0cm 0cm, clip = true]{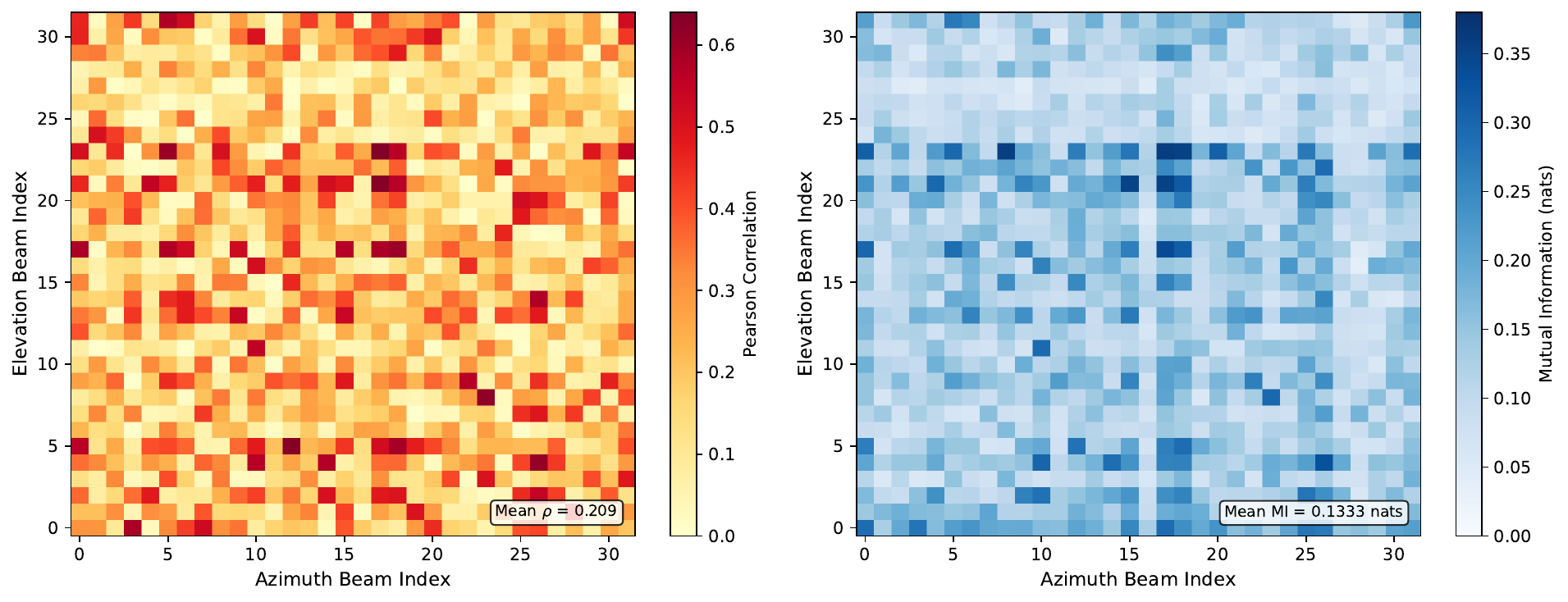}
\caption{\footnotesize{Cross-correlation between the separate azimuth and elevation domain RSRP observations after cross-attention in the O1B\_28 scenario.}}
\label{fig:rsrp_corr_after_o1b28}
\end{figure}

\begin{figure*}[htbp]
\centering
\setlength{\abovecaptionskip}{0cm}
\begin{minipage}[t]{0.48\linewidth}
\centering
\quad
\includegraphics[width=7.6cm, trim = 0cm 0cm 0cm 0cm, clip = true]{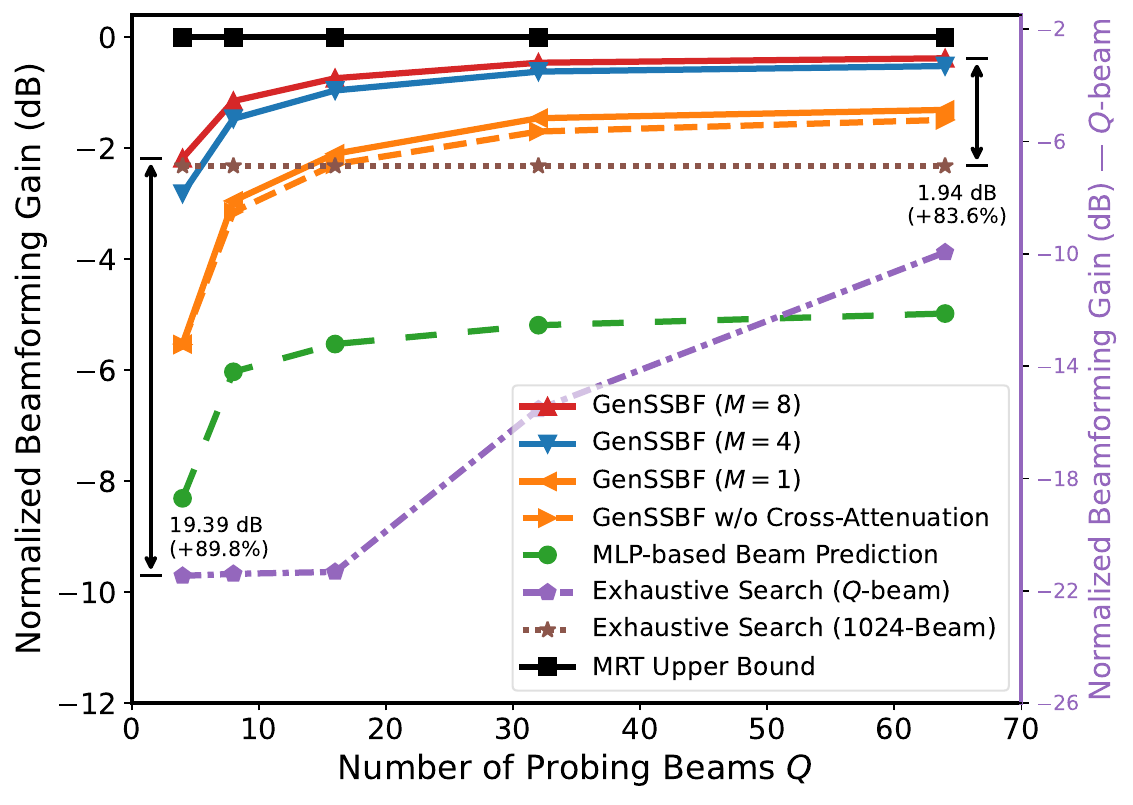}
\caption{\footnotesize{Normalized beamforming gain v.s. the number of probing beams $Q$ in the I2\_28B scenario.}}
\label{fig:I2-28}
\end{minipage}
\quad
\begin{minipage}[t]{0.48\linewidth}
\centering
\includegraphics[width=7cm, trim = 0cm 0cm 0cm 0cm, clip = true]{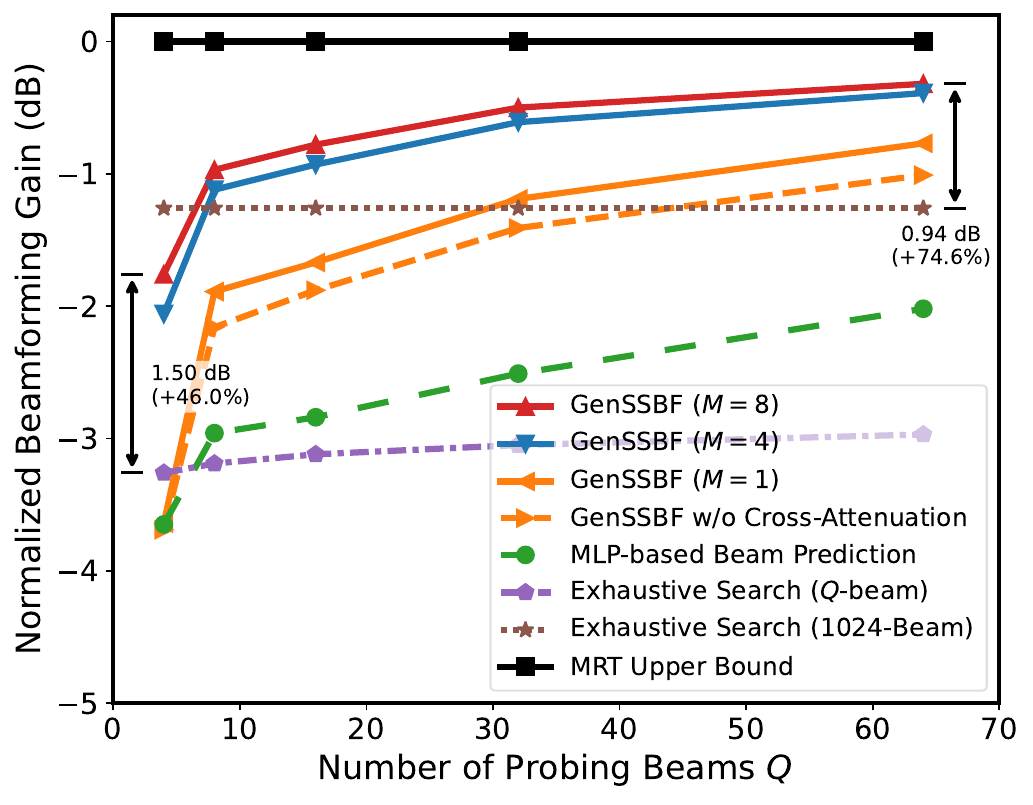}
\caption{\footnotesize{Normalized beamforming gain v.s. the number of probing beams $Q$ in the O1B\_28 scenario.}}
\label{fig:O1B_28}
\end{minipage}
\end{figure*}

\begin{figure}[t]
\begin{center}
\setlength{\abovecaptionskip}{0cm}
\includegraphics[width=7cm, trim = 0cm 0cm 0cm 0cm, clip = true]{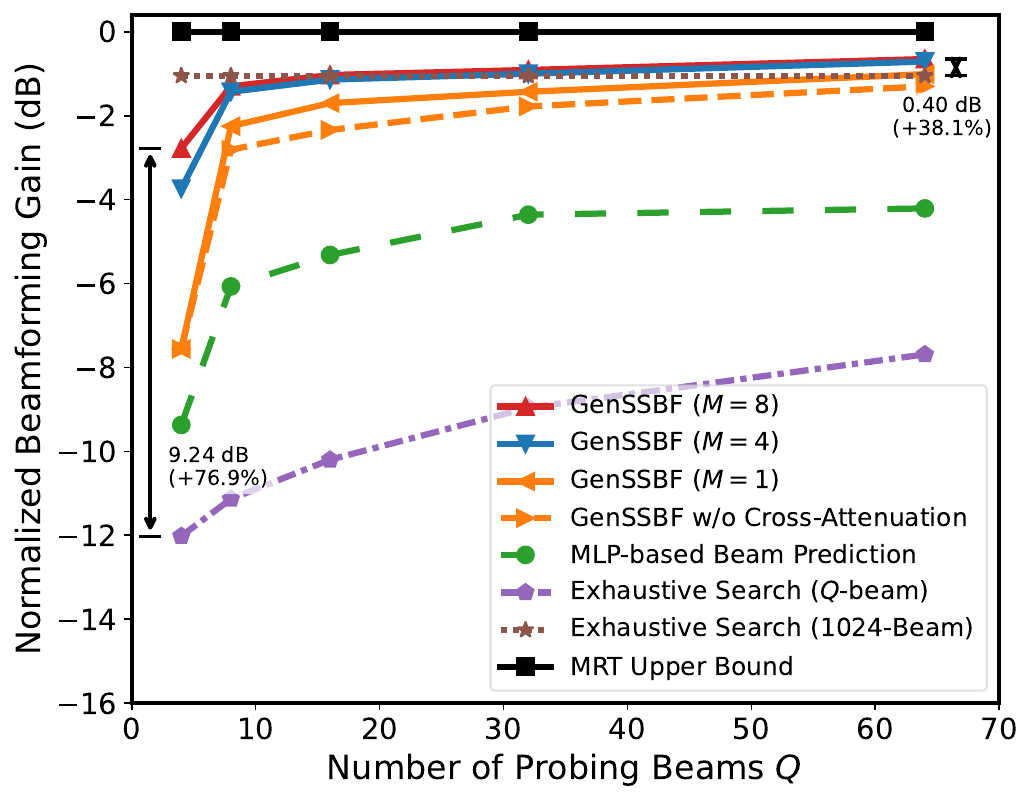}
\caption{\footnotesize{Normalized beamforming gain v.s. the number of probing beams $Q$ in the Boston5G\_28 scenario.}}
\label{fig:Boston5G_28}
\end{center}
\end{figure}

\subsection{Benchmark Schemes}
\label{subsec:baseline}

To evaluate the effectiveness of the proposed GenSSBF framework, we compare it with the following benchmarks.

\begin{itemize}
\item \textbf{GenSSBF w/o Cross-Attention:} This scheme adopts the same overall architecture as the proposed GenSSBF framework, except that the bidirectional cross-attention module is removed. The remaining components, including the decoupled sensing procedure and the conditional flow generator, are kept unchanged. By comparing with this baseline, we evaluate the effectiveness of cross-fusion in recovering the latent azimuth-elevation dependency from separate sensing observations.

\item \textbf{MLP-Based Beam Prediction:} This scheme replaces the conditional normalizing-flow model in GenSSBF with a deterministic multi-layer perceptron, while keeping the same decoupled sensing input and training protocol. By comparing with this baseline, we evaluate the benefit of stochastic candidate generation over direct deterministic beam prediction.

\item \textbf{Exhaustive Search ($Q$-beam):} This scheme adopts the conventional DFT codebook-based beam alignment method, where the BS sweeps the candidate beams and the beam with the largest measured received power is selected. In particular, exhaustive search with $1024$ beams corresponds to full sweeping over the entire DFT codebook. By comparing with this baseline, we evaluate the beamforming gain improvement achieved by the proposed GenSSBF framework over traditional codebook search.

\item \textbf{MRT Upper Bound:} This scheme computes the beamformer using maximum ratio transmission (MRT) with perfect channel knowledge. The resulting performance serves as an ideal upper bound for beamforming gain. By comparing with this baseline, we assess how closely the proposed GenSSBF framework approaches the ideal beamforming performance.
\end{itemize}

\subsection{Performance Evaluation} \label{subsec:performance}




\subsubsection{\textbf{Impact of bidirectional cross-attention on decoupled sensing information loss}}

Fig.~\ref{fig:rsrp_corr_before_o1b28} and Fig.~\ref{fig:rsrp_corr_after_o1b28} illustrate the cross-domain dependency between the azimuth- and elevation-domain sensing observations under the proposed GenSSBF scheme. We consider $32$ sensing beams in each angular dimension. To quantify the dependency between the two sensing branches, we adopt the Pearson correlation coefficient to measure their linear dependence \cite{pearson1895notes}, and the mutual information to characterize their general statistical dependence \cite{shannon1948mathematical, kraskov2004estimating}. 
Fig.~\ref{fig:rsrp_corr_before_o1b28} shows the dependency pattern directly obtained from the decoupled RSRP observations $(\mathbf r_\phi,\mathbf r_\theta)$ before the cross-fused encoder. Since decoupled angular sensing only preserves marginal power responses along the azimuth and elevation dimensions, the resulting cross-domain dependency is relatively weak and fragmented. In contrast, Fig.~\ref{fig:rsrp_corr_after_o1b28} shows that both the Pearson correlation and the mutual information increase after the bidirectional cross-attention module. This indicates that the proposed cross-fused encoder can effectively enhance the latent dependency between the two sensing branches, thereby alleviating the information loss caused by decoupled angular sensing and providing a more informative condition for the subsequent normalizing flow generator.

\begin{table}[t]
\caption{Online Beam Alignment Overhead Comparison}
\label{tab:overhead_comparison}
\centering
\footnotesize
\renewcommand{\arraystretch}{1.15}
\begin{tabular}{|c|c|c|}
\hline
\textbf{Scheme} & \textbf{Overhead} & \textbf{Remark} \\
\hline
Exhaustive Search & $M_{\phi}M_{\theta}$ & Full 2D DFT sweeping \\
\hline
MLP-Based Prediction & $M_{\phi}+M_{\theta}$ & One-shot determination \\
\hline
GenSSBF & $M_{\phi}+M_{\theta}+Q$ & $Q$-candidate generation \\
\hline
\end{tabular}
\end{table}

\subsubsection{\textbf{Performance of beamforming gain}}

Figs.~\ref{fig:I2-28}--\ref{fig:Boston5G_28} show the normalized beamforming gain versus the number of probing beams under different schemes in the I2\_28, O1B\_28, and Boston5G\_28 scenarios. Compared with the benchmark schemes, the proposed GenSSBF ($M=8$) achieves the best beamforming gain performance in all three scenarios.
Specifically, with $64$ probing beams, the proposed GenSSBF ($M=8$) achieves an $83.6\%$, $74.6\%$, and $38.1\%$ gain improvement over Exhaustive Search with the full $1024$-beam DFT codebook in the I2\_28, O1B\_28, and Boston5G\_28 scenarios, respectively, while reducing the sweeping overhead by $93.8\%$. Compared with Exhaustive Search under a very limited probing budget of $Q=4$, the gain improvement of GenSSBF is $89.8\%$ in I2\_28, $46.0\%$ in O1B\_28, and $76.9\%$ in Boston5G\_28.
This is because the proposed framework exploits site-specific channel characteristics to generate a small set of high-fidelity beam candidates from limited separate sensing observations, rather than relying on fixed DFT codebook sweeping. Moreover, GenSSBF ($M=8$) consistently outperforms GenSSBF ($M=1$) in all three scenarios. This is consistent with Proposition~\ref{prop:high-gain} that a larger generation number $M$ increases the probability that the candidate set contains a high-gain beam. 
In addition, the GenSSBF without Cross-Attention scheme achieves worse performance than the full GenSSBF in all scenarios, which confirms the effectiveness of the proposed cross-attention module in extracting the latent correlation between separate sensing observations. 
Furthermore, even GenSSBF ($M=1$) outperforms the MLP-based Beam Prediction scheme, which further demonstrates the advantage of stochastic generative beam synthesis over deterministic beam prediction.

\begin{figure*}[htbp]
\centering
\setlength{\abovecaptionskip}{0cm}

\subfloat[\footnotesize{I2\_28: Exhaustive Search (32-beam)}]{
\includegraphics[width=0.22\linewidth, trim = 0cm 0cm 0cm 0cm, clip = true]{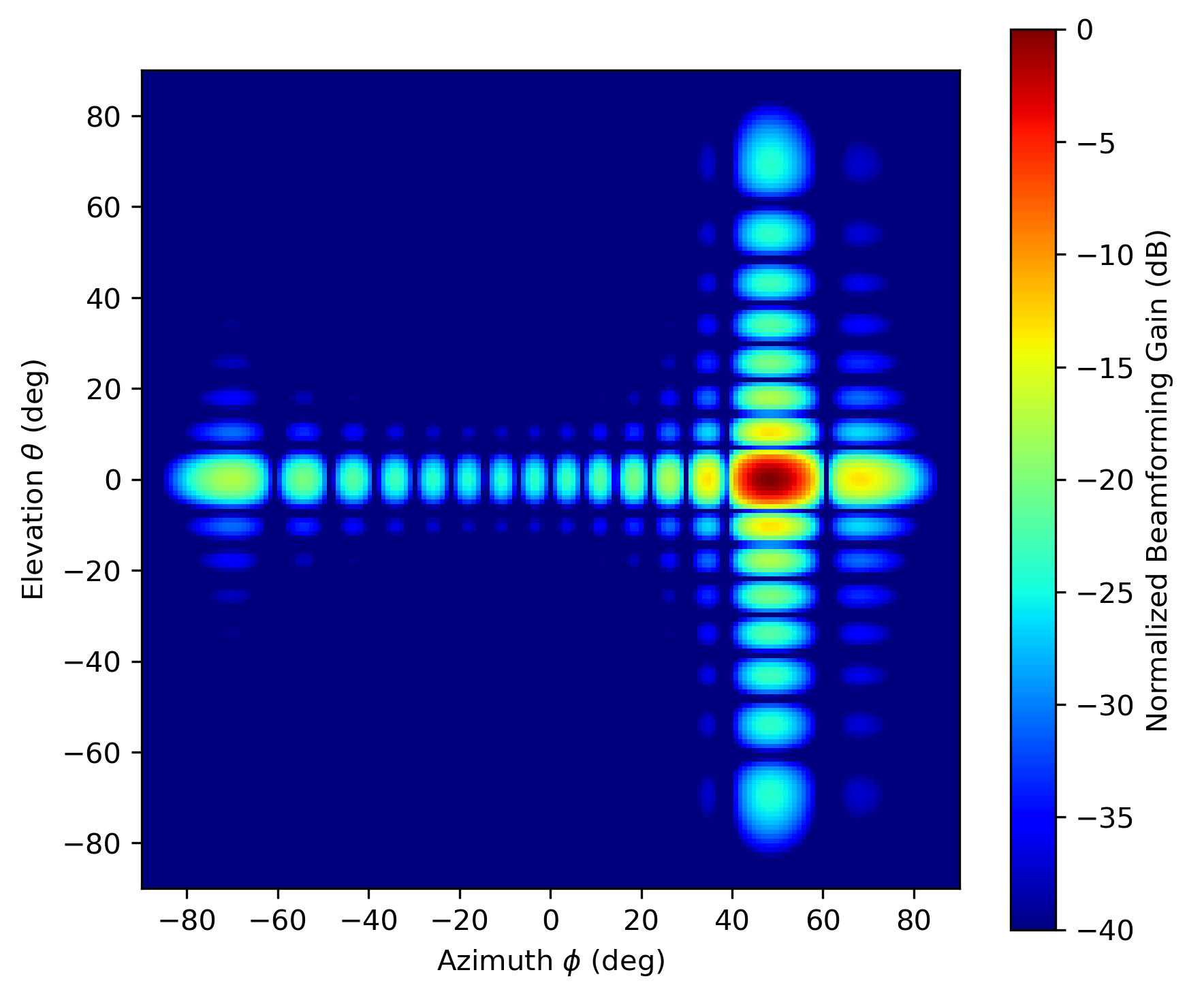}
\label{fig:bp_i2_dft}
}
\hfill
\subfloat[\footnotesize{I2\_28: MLP-based Beam Prediction}]{
\includegraphics[width=0.22\linewidth, trim = 0cm 0cm 0cm 0cm, clip = true]{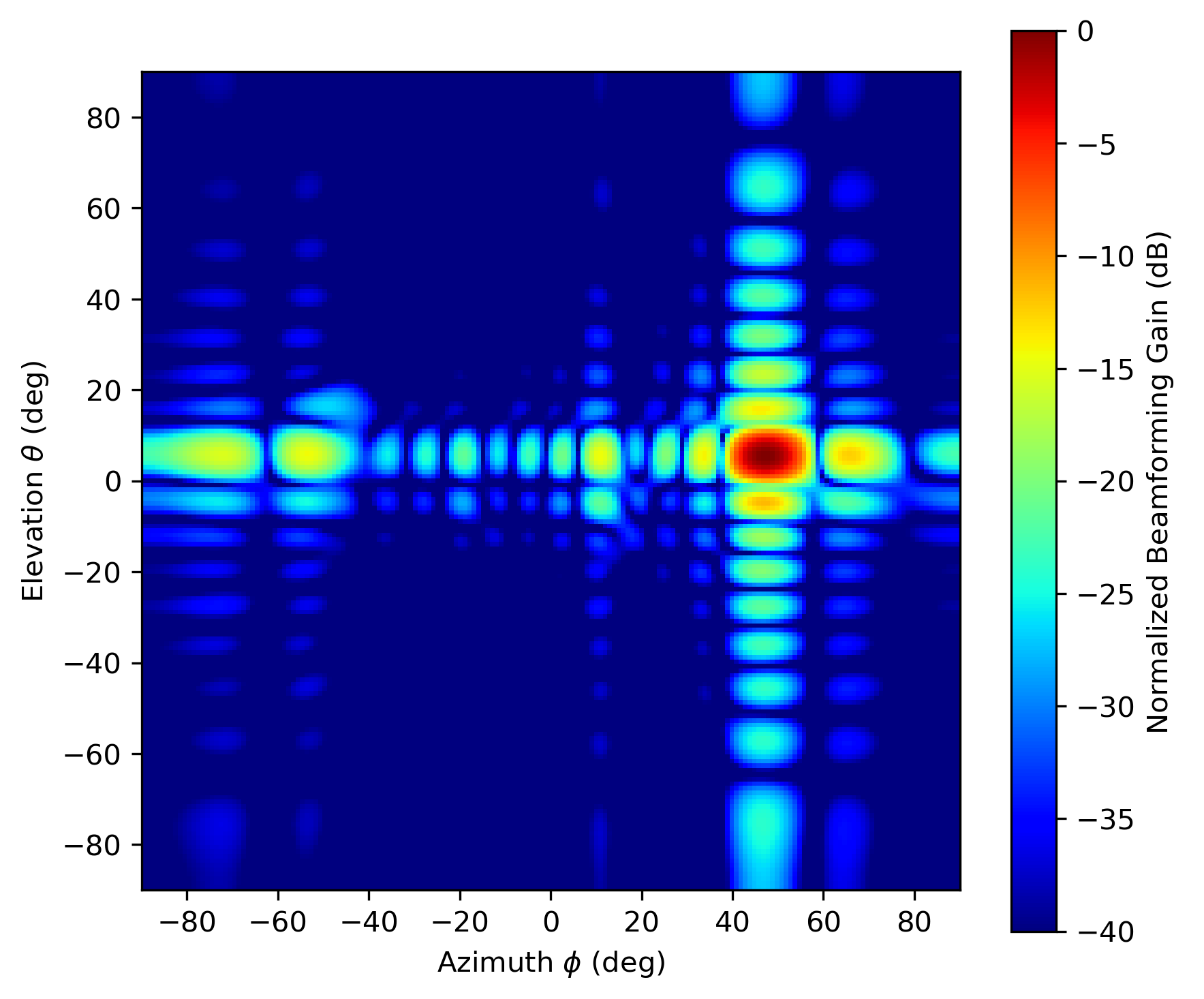}
\label{fig:bp_i2_mlp}
}
\hfill
\subfloat[\footnotesize{I2\_28: GenSSBF ($M=16$)}]{
\includegraphics[width=0.22\linewidth, trim = 0cm 0cm 0cm 0cm, clip = true]{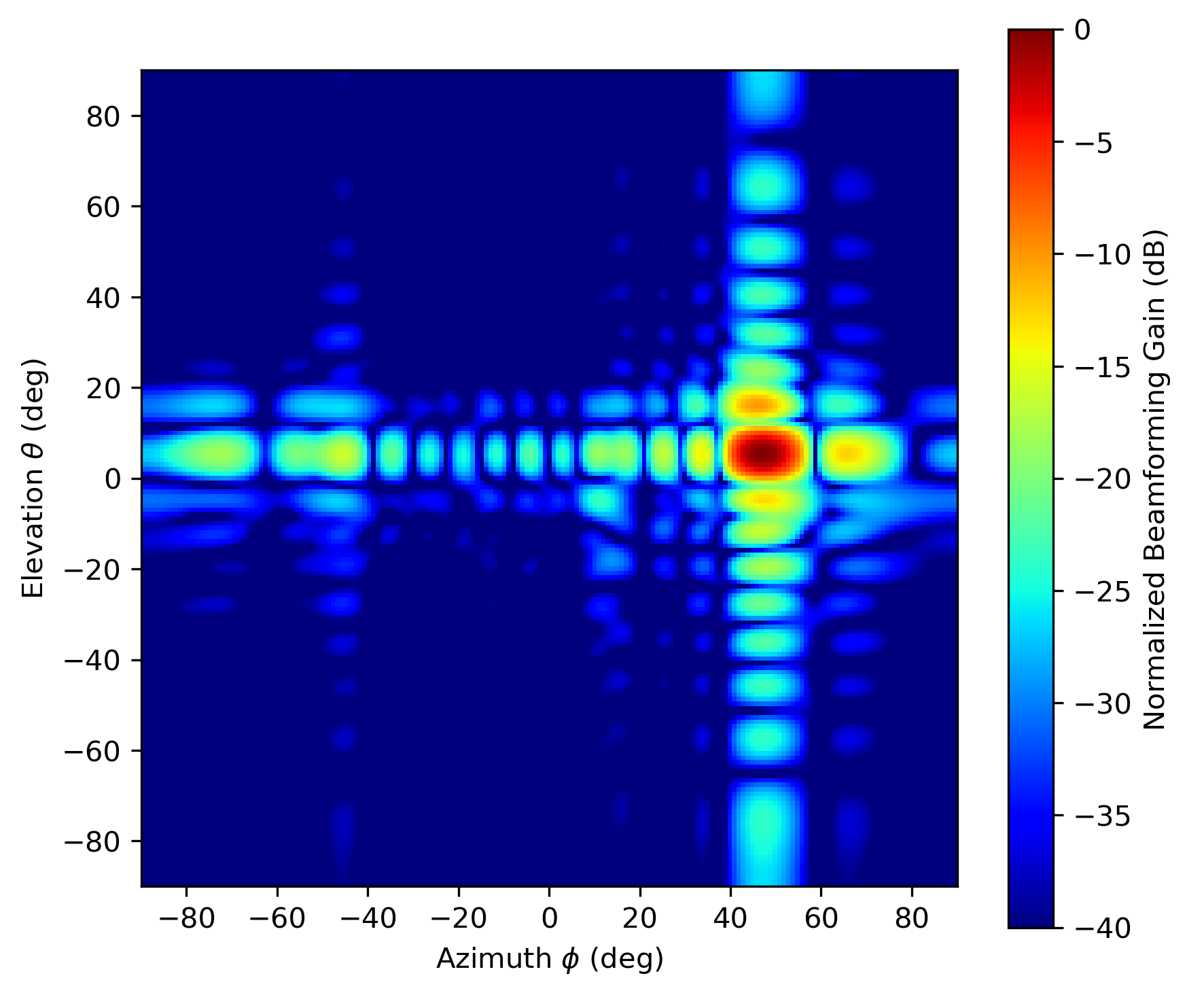}
\label{fig:bp_i2_genssbf}
}
\hfill
\subfloat[\footnotesize{I2\_28: MRT Upper Bound}]{
\includegraphics[width=0.22\linewidth, trim = 0cm 0cm 0cm 0cm, clip = true]{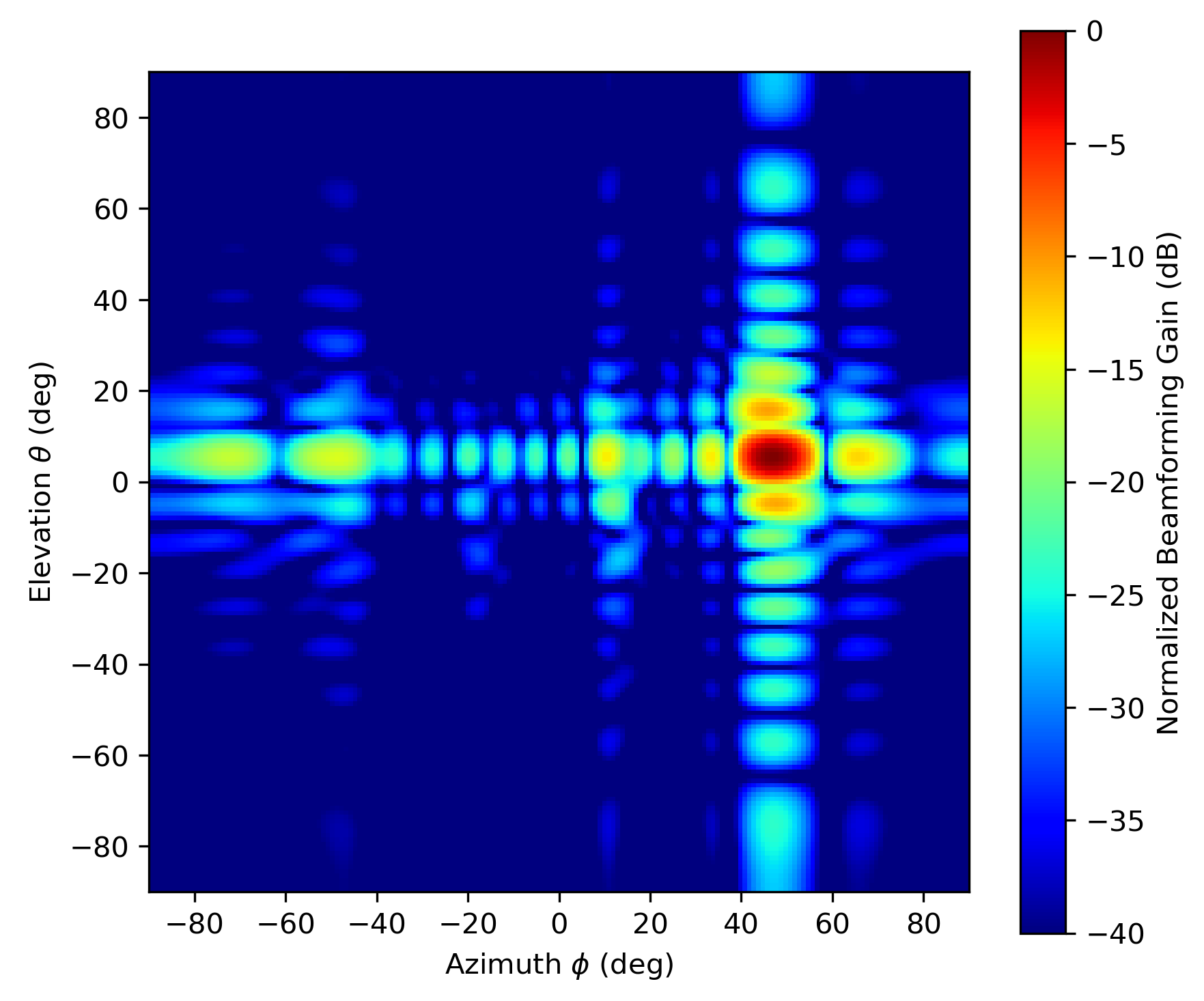}
\label{fig:bp_i2_mrt}
}

\vspace{0.1cm}

\subfloat[\footnotesize{O1B\_28: Exhaustive Search (32-beam)}]{
\includegraphics[width=0.22\linewidth, trim = 0cm 0cm 0cm 0cm, clip = true]{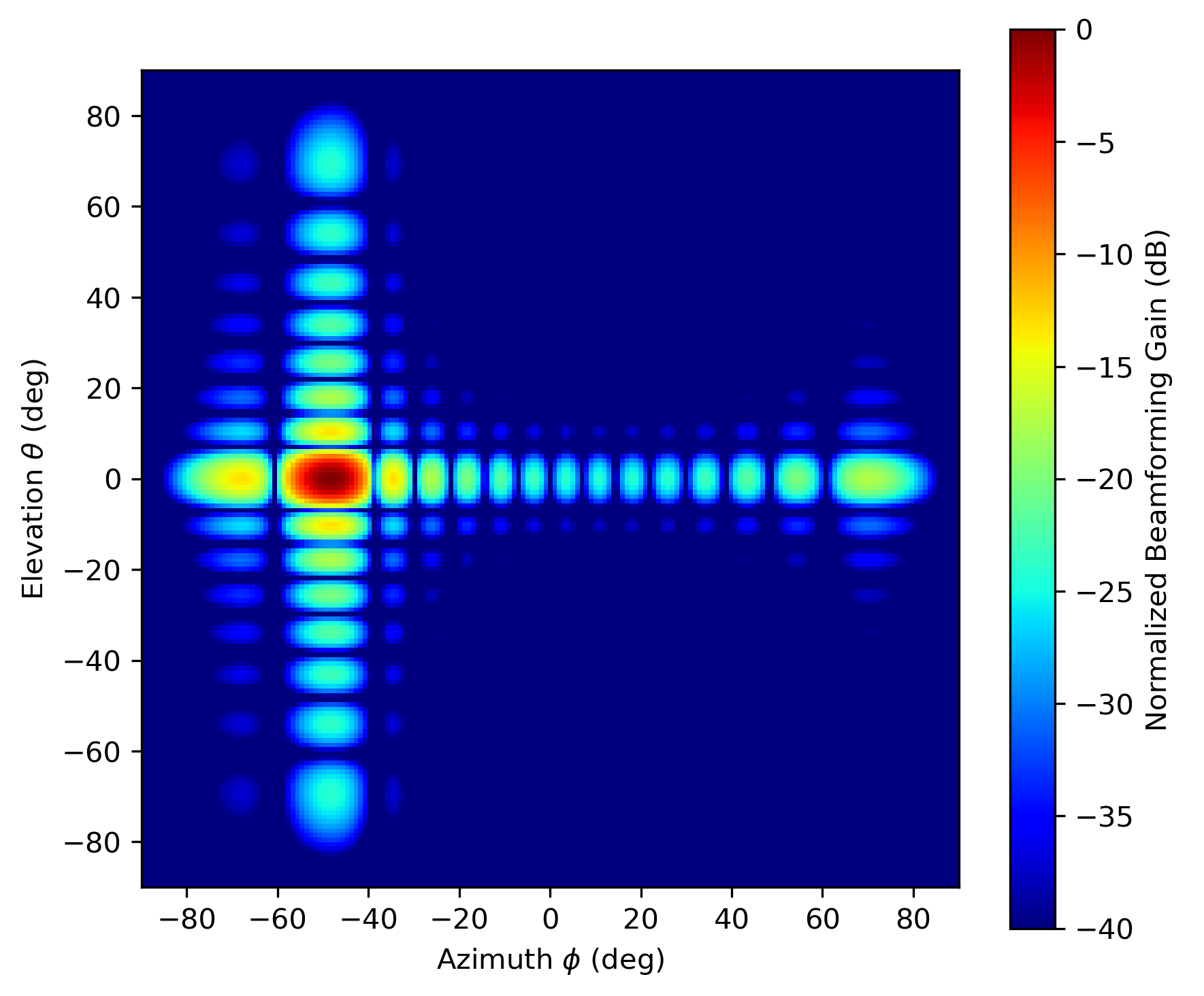}
\label{fig:bp_o1b_dft}
}
\hfill
\subfloat[\footnotesize{O1B\_28: MLP-based Beam Prediction}]{
\includegraphics[width=0.22\linewidth, trim = 0cm 0cm 0cm 0cm, clip = true]{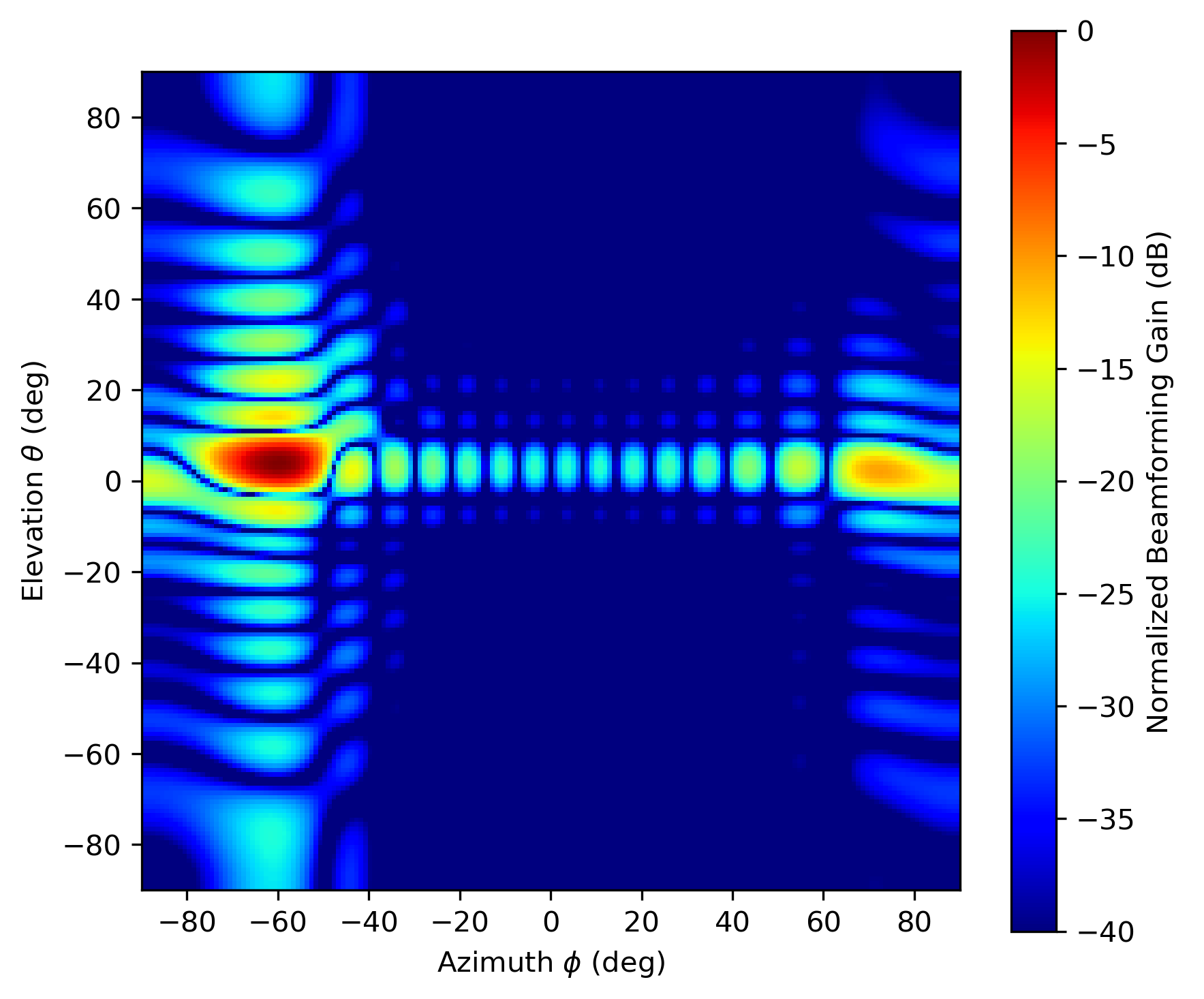}
\label{fig:bp_o1b_mlp}
}
\hfill
\subfloat[\footnotesize{O1B\_28: GenSSBF ($M=16$)}]{
\includegraphics[width=0.22\linewidth, trim = 0cm 0cm 0cm 0cm, clip = true]{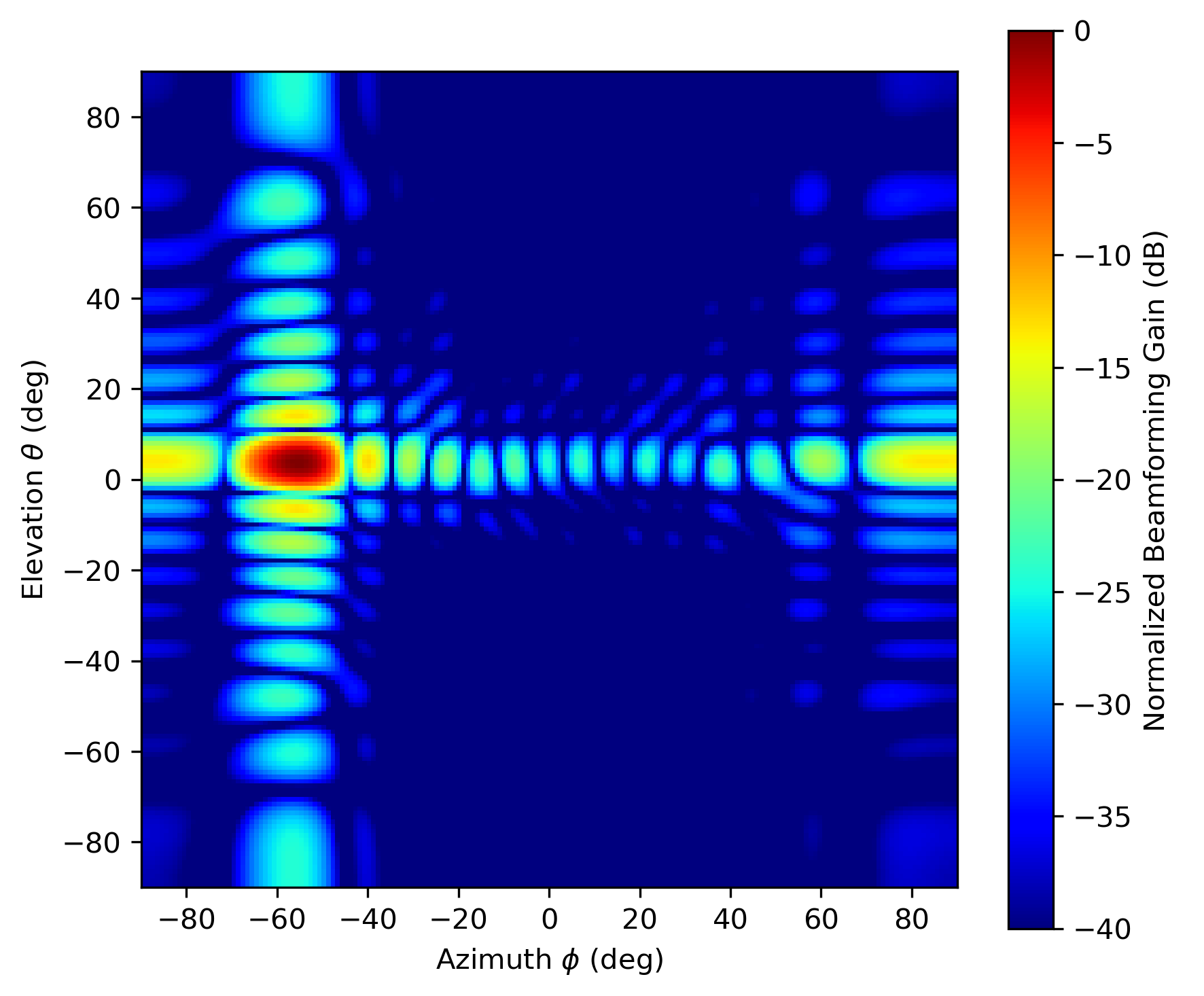}
\label{fig:bp_o1b_genssbf}
}
\hfill
\subfloat[\footnotesize{O1B\_28: MRT Upper Bound}]{
\includegraphics[width=0.22\linewidth, trim = 0cm 0cm 0cm 0cm, clip = true]{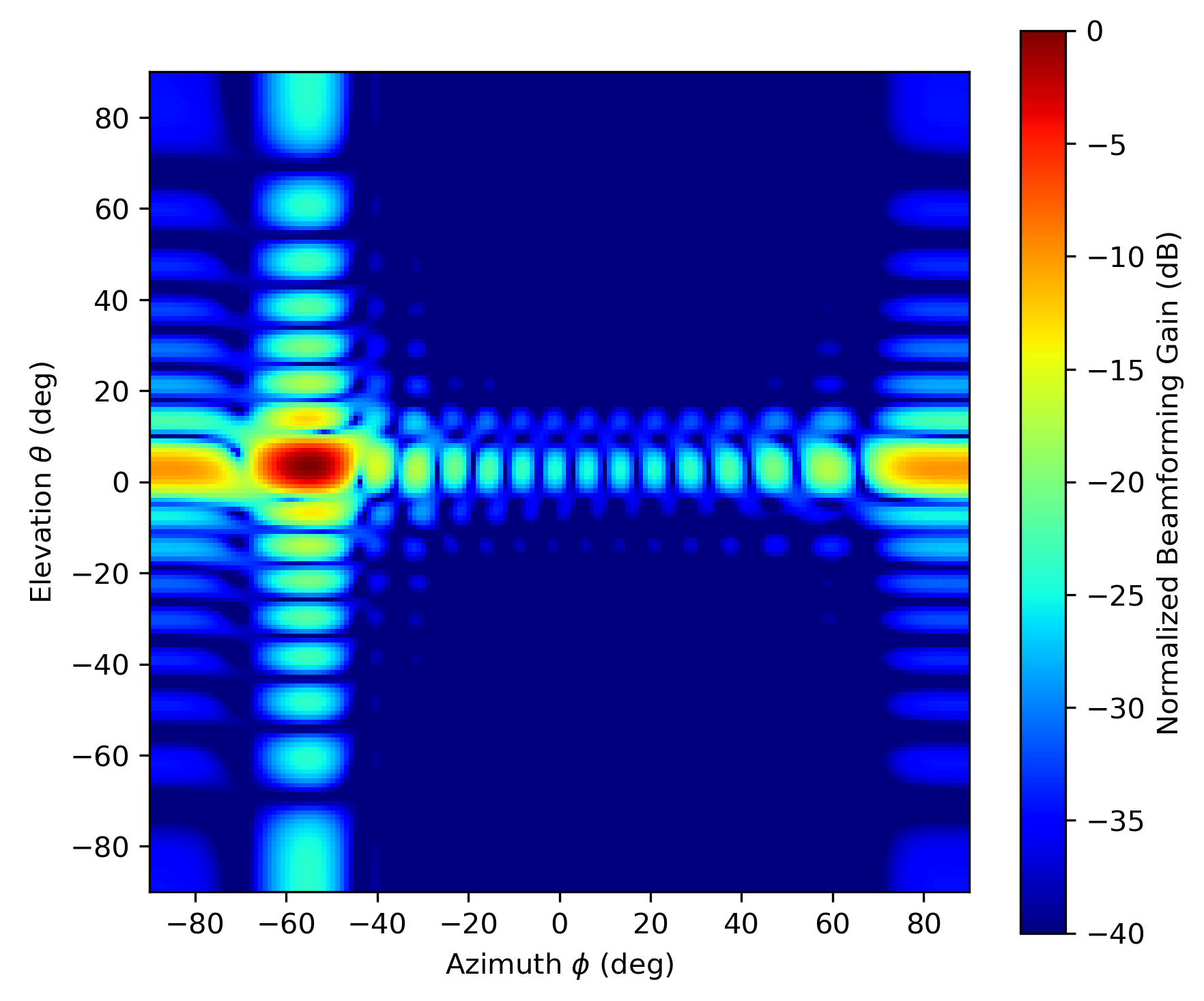}
\label{fig:bp_o1b_mrt}
}

\vspace{0.1cm}

\subfloat[\footnotesize{Boston5G\_28: Exhaustive Search (32-beam)}]{
\includegraphics[width=0.22\linewidth, trim = 0cm 0cm 0cm 0cm, clip = true]{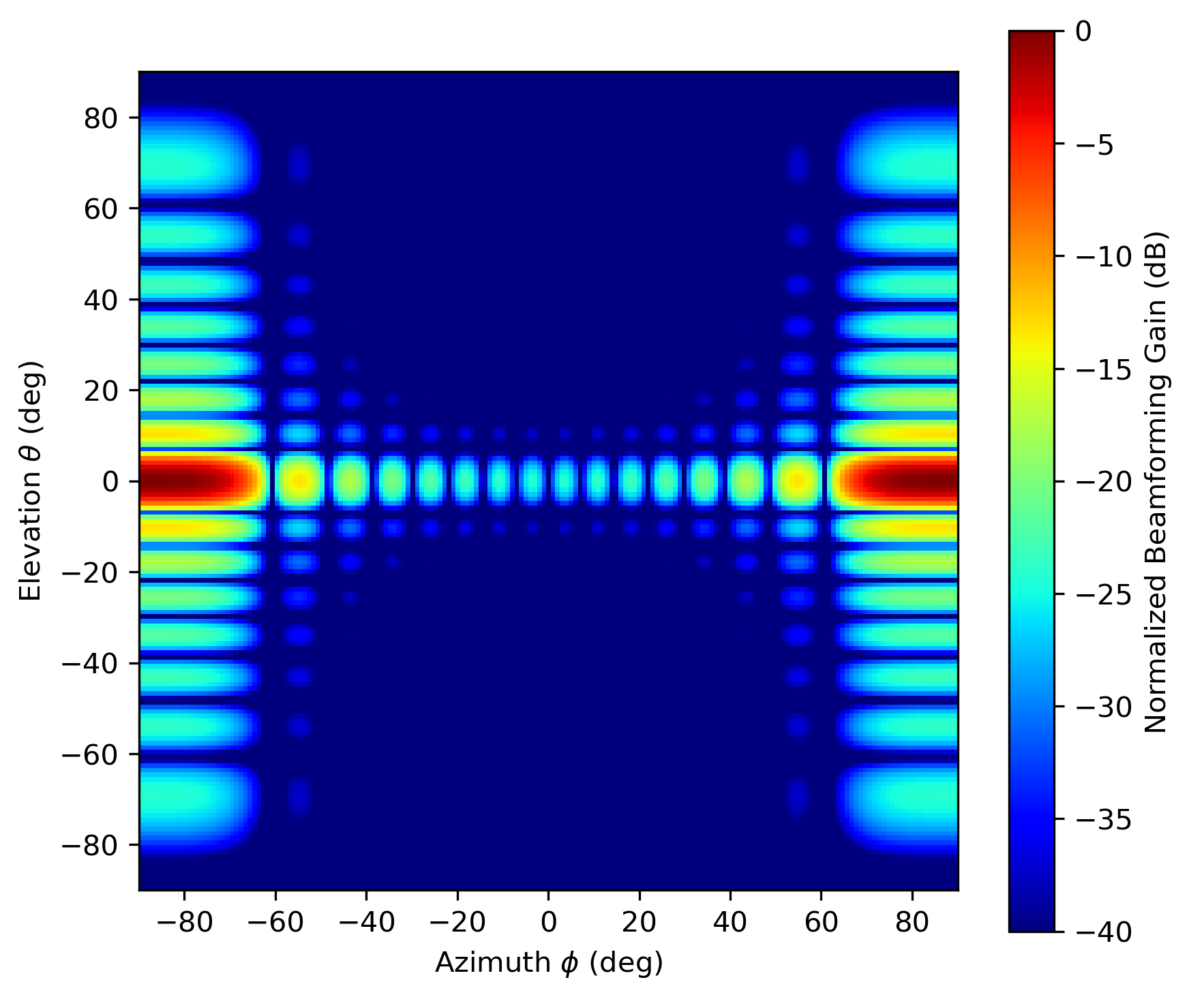}
\label{fig:bp_boston_dft}
}
\hfill
\subfloat[\footnotesize{Boston5G\_28: MLP-based Beam Prediction}]{
\includegraphics[width=0.22\linewidth, trim = 0cm 0cm 0cm 0cm, clip = true]{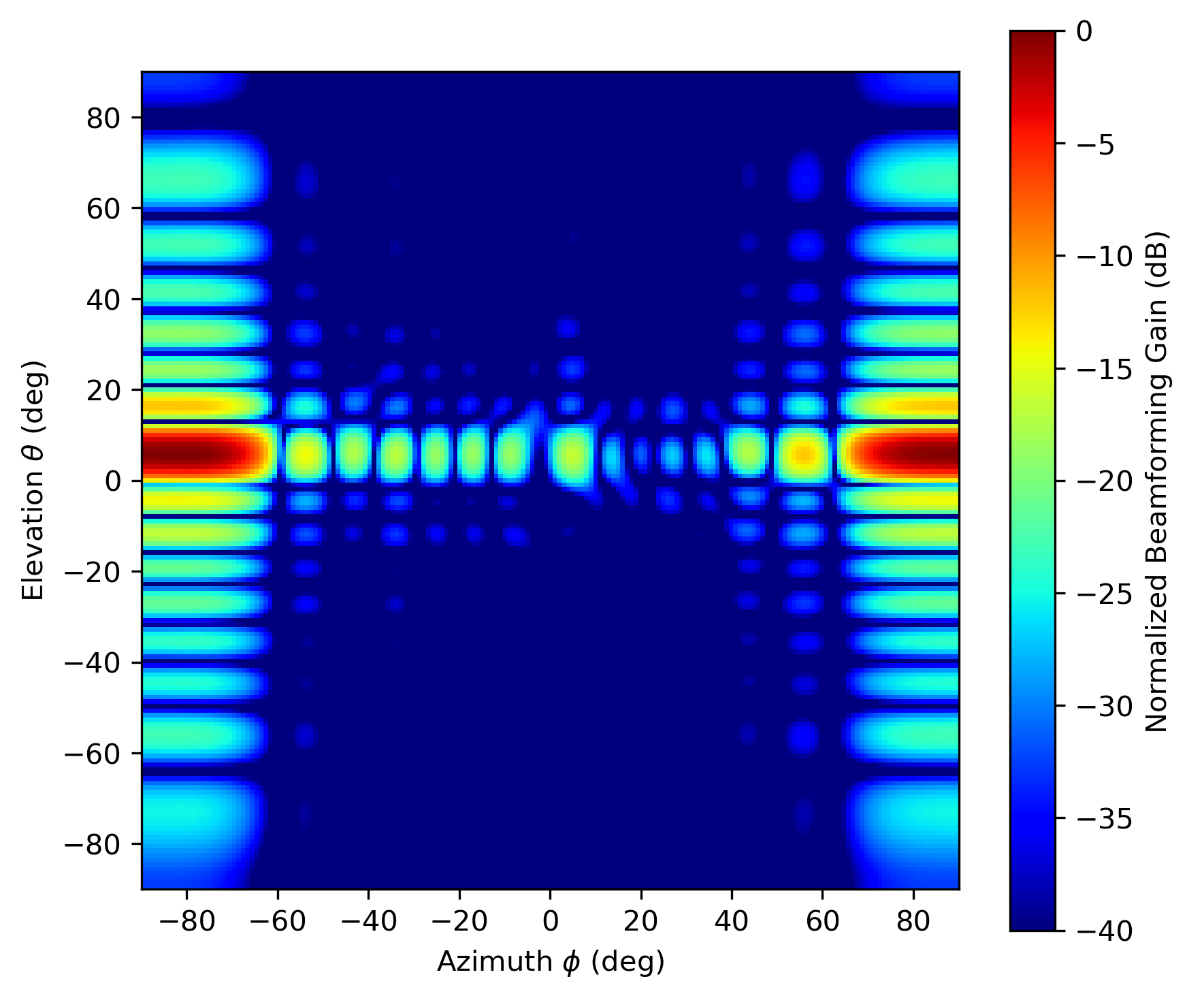}
\label{fig:bp_boston_mlp}
}
\hfill
\subfloat[\footnotesize{Boston5G\_28: GenSSBF ($M=16$)}]{
\includegraphics[width=0.22\linewidth, trim = 0cm 0cm 0cm 0cm, clip = true]{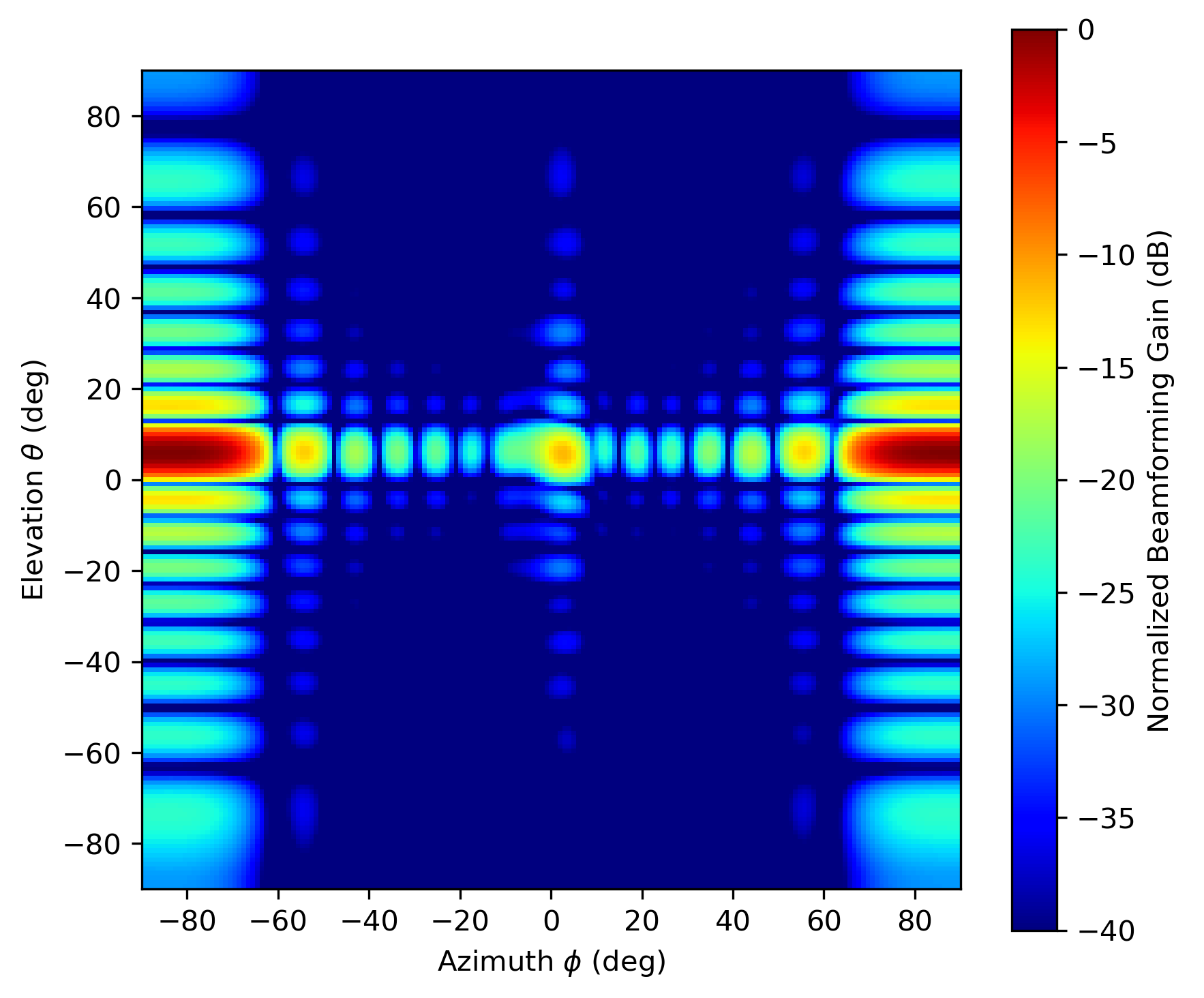}
\label{fig:bp_boston_genssbf}
}
\hfill
\subfloat[\footnotesize{Boston5G\_28: MRT Upper Bound}]{
\includegraphics[width=0.22\linewidth, trim = 0cm 0cm 0cm 0cm, clip = true]{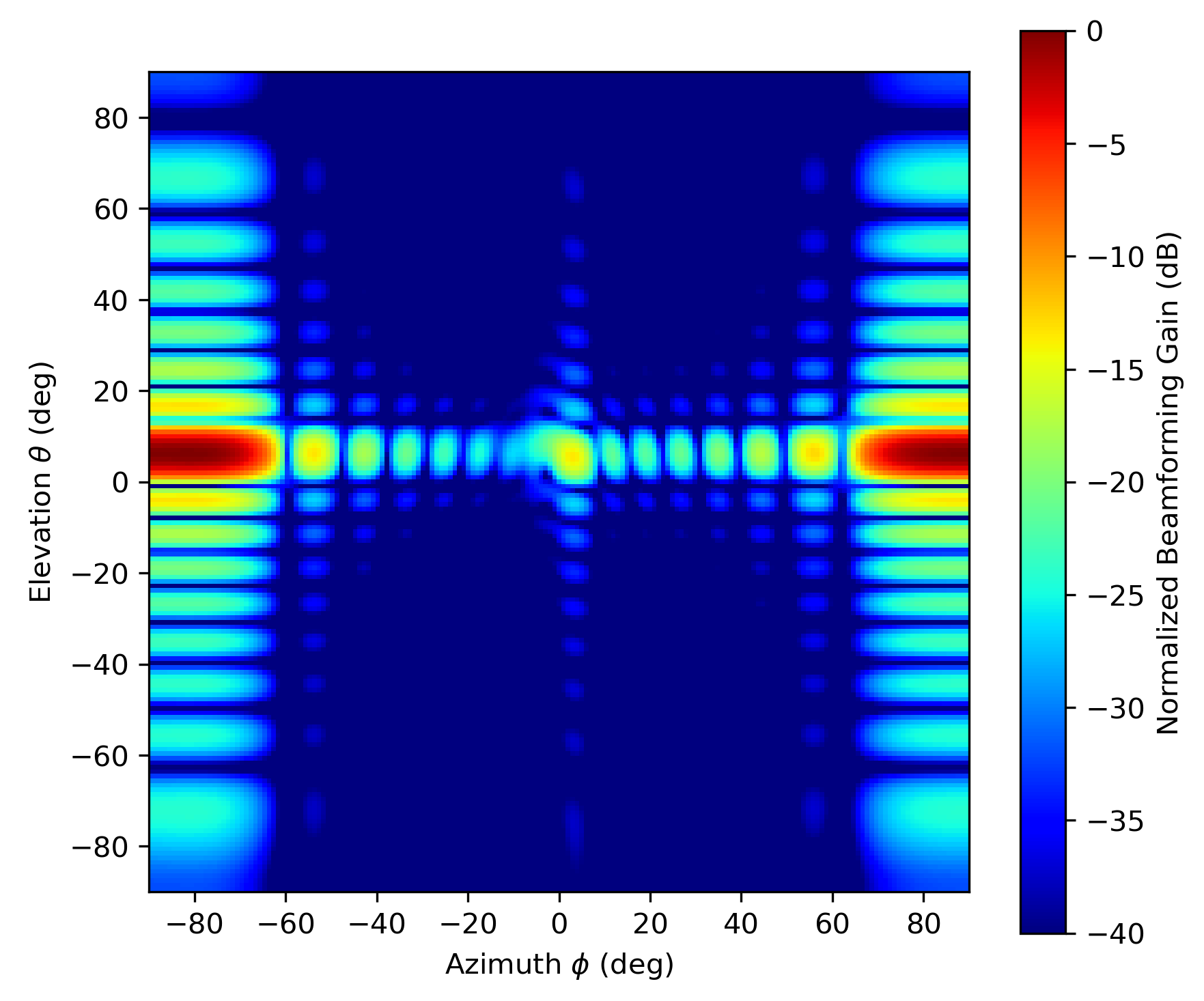}
\label{fig:bp_boston_mrt}
}
\vspace{0.3cm}
\caption{\footnotesize{Beam pattern comparison under different scenarios and beamforming schemes. The first, second, and third rows correspond to the I2\_28, O1B\_28, and Boston5G\_28 scenarios, respectively. In each row, (a)/(e)/(i) show the Exhaustive Search (32-beam), (b)/(f)/(j) show the MLP-based Beam Prediction, (c)/(g)/(k) show the proposed GenSSBF ($M=16$), and (d)/(h)/(l) show the MRT upper bound.}}
\label{fig:beam_pattern_all}
\end{figure*}

\subsubsection{\textbf{Overhead analysis}}

We compare the online beam alignment overhead under different schemes in Table~\ref{tab:overhead_comparison}. Compared with Exhaustive Search, the proposed GenSSBF and the MLP-based Beam Prediction scheme both require only separate azimuth and elevation sensing, and therefore reduce the probing overhead from $M_{\phi}M_{\theta}$ to $M_{\phi}+M_{\theta}$. This substantially lowers the online sensing cost, especially when the beam resolution in both angular dimensions is high.
In the MLP-based scheme, the beamformer is predicted in a one-shot deterministic manner, and its total online overhead is $M_{\phi}+M_{\theta}$. 
In contrast, GenSSBF generates a small set of candidate beamformers, and selects the best one through short pilot-based verification, which results in a total overhead of $M_{\phi}+M_{\theta}+Q$. Although this introduces an extra overhead of $Q$, the value of $Q$ is typically much smaller $Q \ll M_{\phi}M_{\theta}$. Therefore, the overhead remains significantly lower than that of full DFT sweeping.
The overhead difference reflects the tradeoff between robustness and online cost. 
Exhaustive Search scheme incurs the largest overhead because it requires full two-dimensional beam sweeping. 
The MLP-based Beam Prediction scheme achieves the lowest overhead. However, due to the ambiguity inherent in decoupled sensing, its one-shot deterministic prediction is generally insufficient to reliably obtain a high-quality beamformer. 
The proposed GenSSBF introduces only a small additional verification cost over the MLP-based scheme, while significantly improving beamforming performance through stochastic generation. Therefore, GenSSBF achieves a favorable balance between beam alignment accuracy and online sensing overhead.

\begin{figure*}[htbp]
\centering
\setlength{\abovecaptionskip}{0cm}
\begin{minipage}[t]{0.48\linewidth}
\centering
\includegraphics[width=7cm, trim = 0cm 0cm 0cm 0cm, clip = true]{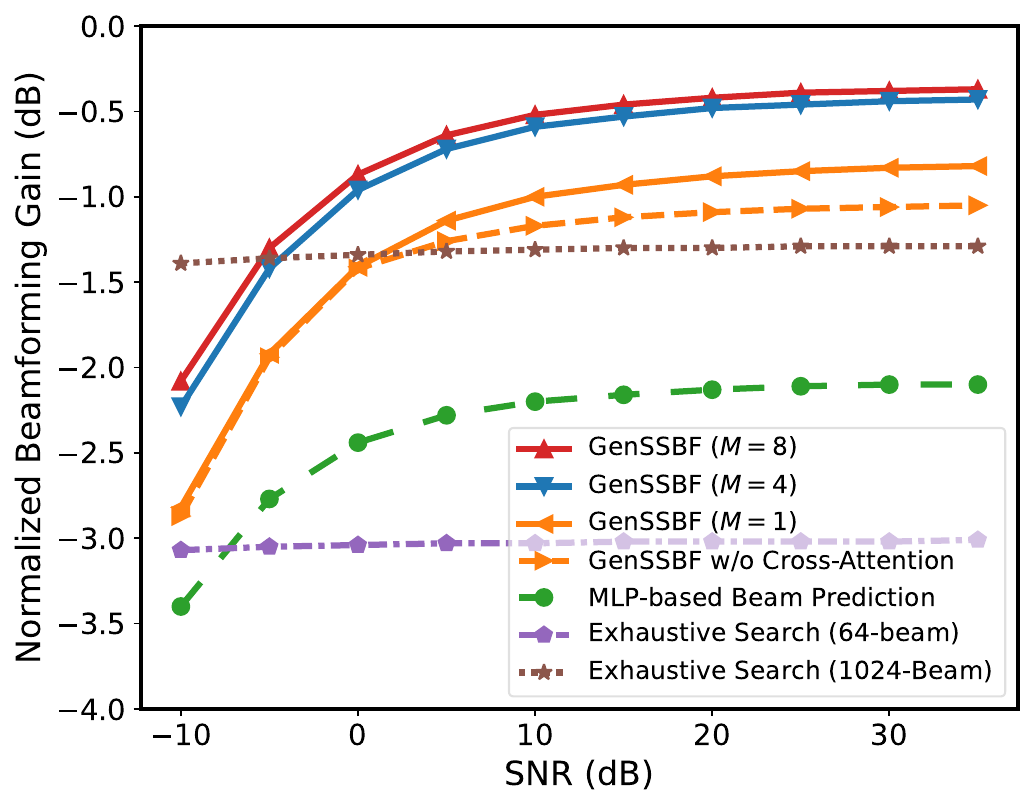}
\caption{\footnotesize{Normalized beamforming gain v.s. SNR in the I2-28B scenario.}}
\label{fig:SNR_I2_28B}
\end{minipage}
\quad
\begin{minipage}[t]{0.48\linewidth}
\centering
\includegraphics[width=7cm, trim = 0cm 0cm 0cm 0cm, clip = true]{./snr_beamforming_gain_o1b_28.pdf}
\caption{\footnotesize{Normalized beamforming gain v.s. SNR in the O1B-28 scenario.}}
\label{fig:SNR_O1B_28}
\end{minipage}
\end{figure*}

\begin{figure}[t]
\begin{center}
\setlength{\abovecaptionskip}{0cm}
\includegraphics[width=7cm, trim = 0cm 0cm 0cm 0cm, clip = true]{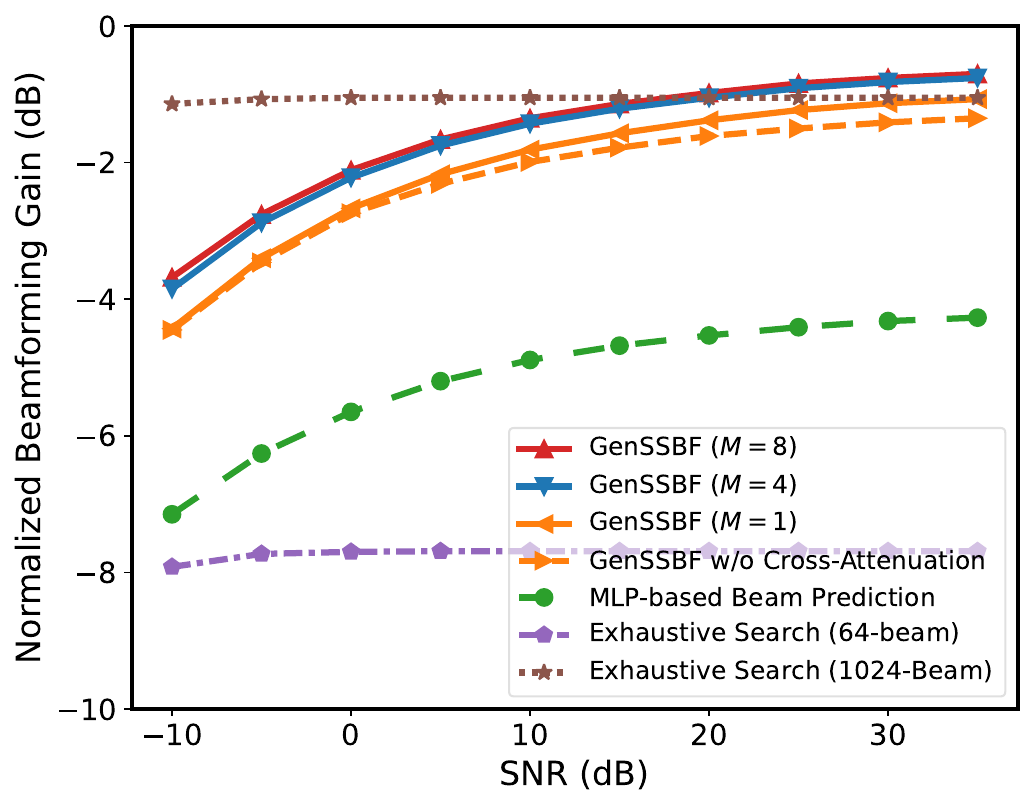}
\caption{\footnotesize{Normalized beamforming gain v.s. SNR in the Boston5G-28 scenario.}}
\label{fig:SNR_5G_28}
\end{center}
\end{figure}

\subsubsection{\textbf{Visualization of beam patterns}}

Fig.~\ref{fig:beam_pattern_all} shows the beam patterns under different schemes in the I2\_28, O1B\_28, and Boston5G\_28 scenarios. Under the same number of probing beams $32$, the proposed GenSSBF achieves beam patterns that are consistently closer to the MRT Upper Bound than the Exhaustive Search and the MLP-based Beam Prediction schemes. This observation is consistent with the beamforming gain results in Figs.~\ref{fig:I2-28}--\ref{fig:Boston5G_28}.
The reason is that the DFT-based Exhaustive Search scheme is restricted to a fixed codebook and thus cannot adapt to the site-specific propagation structure, which leads to evident beam mismatch in all three scenarios. By contrast, the proposed GenSSBF generates beamformers directly from the separate sensing observations and can better align the main lobe with the dominant propagation direction. Therefore, its beam pattern is much closer to that of MRT, which is the ideal benchmark with perfect channel knowledge.
Moreover, compared with the MLP-based beam prediction scheme, the proposed GenSSBF also achieves a beam pattern closer to MRT. This is because the deterministic MLP tends to produce a compromised solution under ambiguous decoupled sensing, whereas the stochastic generative framework can better capture the plausible high-gain beam directions. Therefore, the visualization results further verify the advantage of the proposed GenSSBF in synthesizing high-fidelity site-specific beamformers under limited probing overhead.

\subsubsection{\textbf{Robustness to noisy sensing}}

Figs.~\ref{fig:SNR_I2_28B}--\ref{fig:SNR_5G_28} show the normalized beamforming gain versus the SNR of the received RSRP measurements in the I2\_28, O1B\_28, and Boston5G\_28 scenarios, where the number of probing beams is fixed to $Q=64$. The noise power is determined according to the median received power of the probing DFT beams and the target SNR, which varies from $-10$ dB to $35$ dB.
We can observe that the beamforming gain of all schemes improves as the SNR increases. This is because a lower SNR introduces stronger uncertainty into the separate sensing observations, making the RSRP vectors less informative for beam prediction. In contrast, when the SNR is high, the sensed RSRP better reflects the underlying angular structure of the channel, and thus provides more reliable guidance for beam generation.
Compared with the benchmark schemes, the proposed GenSSBF achieves superior beamforming gain over a wide SNR range in all three scenarios. For example, in the I2\_28 scenario, GenSSBF outperforms Exhaustive Search with the full $1024$-beam DFT codebook by $83.6\%$ at $35$ dB SNR. Similar performance advantages are observed in the O1B\_28 and Boston5G\_28 scenarios, where the gain improvement over Exhaustive Search with $1024$ beams reaches $71.3\%$ and $33.3\%$ at $35$ dB SNR, respectively. Moreover, under the same probing beams, GenSSBF consistently outperforms Exhaustive Search with $64$ beams, which demonstrates the effectiveness of the proposed generative beam prediction framework under practical low-overhead sensing.
It is also observed that the performance of GenSSBF degrades at very low SNR. When the SNR falls below approximately $-5$ dB, the separate RSRP observations become heavily contaminated by noise, and the resulting lack of reliable sensing information makes accurate beam generation more difficult. 
In this case, exhaustive search with the full DFT codebook may achieve higher beamforming gain, at the expense of substantially larger sweeping overhead. Therefore, the results indicate that the proposed GenSSBF is robust to sensing noise in the practical SNR region, while retaining a favorable gain-overhead tradeoff compared with conventional codebook-based beam sweeping.


\section{Conclusions} \label{sec:conclusion}

In this paper, we investigated low-overhead generative site-specific beamforming for UPA-enabled mmWave downlink systems. Specifically, we proposed a decoupled sensing strategy, where the BS probes the azimuth and elevation domains independently to reduce the online beam sweeping overhead from multiplicative complexity to linear complexity. We showed that such decoupled channel sensing introduces intrinsic ambiguity by discarding the joint azimuth-elevation coupling of the UPA channel. To address this issue, we developed a cross-fused GenSSBF framework, where a bidirectional cross-attention encoder extracts the latent dependency between the separate RSRP observations, and a conditional normalizing flow generates multiple candidate beamformers. Then, we designed a task-oriented training objective to encourage the generated candidate set to contain at least one high-gain beam. 
Simulation results based on the realistic propagation scenarios showed that the proposed framework outperforms deterministic beam prediction, GenSSBF without cross-attention, and conventional DFT codebook search in terms of normalized beamforming gain.
Moreover, the proposed scheme maintained robust performance under practical noisy sensing conditions and generated beam patterns close to the MRT upper bound with much lower overhead than full DFT codebook sweeping.




\balance
\bibliographystyle{IEEEtran}
\bibliography{reference/mybib}

\end{document}